\documentclass[30pt]{amsart}
\usepackage{epsfig, amssymb}
\usepackage{color}

\usepackage[latin1]{inputenc}

\newtheorem{theorem}{Theorem}[section]
\newtheorem{proposition}[theorem]{Proposition}
\newtheorem{lemma}[theorem]{Lemma}
\newtheorem{corollary}[theorem]{Corollary}

\theoremstyle{definition}
\newtheorem{definition}[theorem]{Definition}

\theoremstyle{remark}
\newtheorem{remark}[theorem]{Remark}

\numberwithin{equation}{section}

\newcommand{\be}{\begin{equation}}
\newcommand{\ee}{\end{equation}}

\newcommand{\bbC}{{\mathbb C}}
\newcommand{\bbN}{{\mathbb N}}
\newcommand{\bbZ}{{\mathbb Z}}
\newcommand{\bbR}{{\mathbb R}}

\newcommand{\bbS}{{\mathbb S}}

\newcommand{\calQ}{{\mathcal Q}}
\newcommand{\calK}{{\mathcal K}}

\newcommand{\calM}{{\mathcal M}}
\newcommand{\calP}{{\mathcal P}}

\newcommand{\calC}{{\mathcal C}}

\newcommand{\calS}{{\mathcal S}}

\newcommand{\calO}{{\mathcal O}}

\newcommand{\vp}{\varphi}

\newcommand{\R}{\mathbb R}

\newcommand{\norm}[1]{\lVert#1\rVert}

\newcommand{\w}{ P }
\newcommand{\uv}{ \mathcal{R} }

\newcommand{\U}[1]{e^{\frac{i\widetilde{W}_{\leq #1}}{\hbar}}}
\newcommand{\Uinv}[1]{e^{\frac{-i\widetilde{W}_{\leq #1}}{\hbar}}}
\newcommand{\W}[1]{\widetilde{W}_{\leq #1}}
\newcommand{\bz}{\bar{z}}
\usepackage{stmaryrd} 
\newcommand{\normleft}{\left|\left|}
\newcommand{\normright}{\right|\right|}
\newcommand{\mul }{\langle\mu,\nu|}
\newcommand{\mur}{|\mu,\nu\rangle}
\newcommand{\opH}{H(x,\hbar D_x)}
\newcommand{\tendvers}[2]{\underset{#1\rightarrow #2}{\longrightarrow}}
\newcommand{\bcb}{\begin{color}{blue}}
\newcommand{\bcr}{\begin{color}{red}}
\newcommand{\ecc}{\end{color}}
\newcommand{\tcr}{}

\newcommand{\tcrr}{}
\newcommand{\tcb}{}
\title[recovering Hamiltonians]{Recovering the Hamiltonian from spectral data}

\author[C. H\'eriveaux and T. Paul ]{C. H\'eriveaux and T. Paul}
\address{CMLS \'Ecole polytechnique, 91 128 Palaiseau cedex}
\email{cyrille.heriveaux@math.polytechnique.fr}
\address{CNRS and CMLS \'Ecole polytechnique, 91 128 Palaiseau cedex}
\email{paul@math.polytechnique.fr}

\begin{document}

\maketitle


\begin{abstract}
We show that the contributions to the Gutzwiller formula with observable associated to the iterates 
of a given elliptic nondegenerate periodic trajectory $\gamma$ and to certain families of observables localized near $\gamma$ determine the quantum Hamiltonian in a formal neighborhood of the trajectory $\gamma$, that is the full Taylor expansion of its total symbol near $\gamma$. We also treat the ``bottom of a well" case both for general and Schrödinger operators, and give some analog  classical results.

\end{abstract}

\tableofcontents

\section{Introduction and main results}
It is well known that spectral properties of semiclassical Hamiltonians and dynamical properties of their principal symbols are linked.
Even when there is no precise information ``eigenvalue by eigenvalue" of the spectrum, the so-called Gutzwiller trace formula provide information on averages of the spectrum at scale of the Planck constant.
More precisely, let $H(x,\hbar D_x)$ be a self-adjoint
semiclassical elliptic pseudodifferential operator on a compact manifold $X$ of dimension $n+1$, whose symbol $H(x,\xi)$ is proper (as a map from $T^*X$ into $\bbR$). 
Let E be a regular value of $H$ and $\gamma$ a non-degenerate periodic trajectory of \tcr{primitive} period 
$T_\gamma$ lying on the energy surface $H=E$.

\tcr{Consider the Gutzwiller trace (see \cite{gut})
\be\label{bom}
\text{Tr}\left(\psi_r\left(\frac{H(x,\hbar D_x)-E}\hbar\right)\right)=\sum_i \psi_r\left(\frac{E-E_i}\hbar\right)
\ee
where for $r\in\bbZ^*$, $\psi_r$ is a $C^\infty$ function whose Fourier transform is compactly supported with 
support in a small enough neighborhood of $rT_\gamma$ and is identically one in a 
still smaller neighborhood containing $rT_\gamma$. As shown in \cite{pu1}, \cite{pu2} (\ref{bom}) has an asymptotic 
expansion
\be\label{boj}
e^{i\frac{S_\gamma}\hbar + \sigma_\gamma}\sum_{k=0}^\infty a_k^r\hbar^k
\ee
}

In \cite{vgtp} was shown how to compute the terms of this expansion to 
\underline{all}\rm \ orders in terms of a microlocal Birkhoff canonical form for $H$ in a formal neighborhood of $\gamma$, and that 
the family of constants $(a_k^r)_{(k,r)\in\bbN\times\bbZ^*}$ determine the microlocal 
Birkhoff canonical form for $H$ in a formal neighborhood of $\gamma$ (and hence, a 
fortiori, determine the classical Birkhoff canonical form).
When it is known ``a priori" that $H(x,\hbar D_x)$ is a Schr\"odinger operator, it is known that the normal form near the bottom of a well determines part of the potential $V$ \cite{gur}. But in the general case the Gutzwiller formula will determine only the normal form of the Hamiltonian, that is to say $H(x,\hbar D_x)$  only modulo unitary operators, and its principal symbol only modulo symplectomorphisms. Of course it cannot determine more, as the spectrum, and a fortiori the trace, is insensitive to unitary conjugation. 
The aim of this paper is to address the question of determining the \textit{true} Hamiltonian from more precise spectral data, namely from the Gutzwiller trace formula with observables.

It is well know that, for any pseudodifferential operator $O(x,\hbar D_x)$ of symbol $\mathcal O(x,\xi)$, there is a  result equivalent  to \eqref{boj} for the following quantity 
\be\label{gutzo}
\text{Tr}\left(O(x,\hbar D_x)\psi\left(\frac{H(x,\hbar D_x)-E}\hbar\right)\right)=\sum_i \langle\vp_j,O(x,\hbar D_x)\vp_j\rangle\psi\left(\frac{E-E_i}\hbar\right),
\ee
(here $\vp_j$ is meant as the eigenvector of eigenvalue $E_j$)
under the form of an asymptotic expansion of the form
\be\label{Sbojo}
e^{i\frac{S_\gamma}\hbar + \sigma_\gamma}\sum_{k=0}^\infty a_k^r(\mathcal O)\hbar^k
\ee
where $a_k^r$ are distributions supported on $\gamma$.

We will show in the present paper that the knowledge of the coefficients $a_k^\gamma(O)$ for $O$ belonging to some family of  observables  localized near $\gamma$ is enough to determine the full Taylor expansion of the total symbol of $H(x,\hbar D_x)$ near $\gamma$,  in other words $H(x,\hbar D_x)$ microlocally in a formal neighborhood of $\gamma$, when $\gamma$ is non-degenerate elliptic 
(including the case where $\gamma$ is reduced to a point (bottom of a well)). 
\tcb{
Let us first remark that 
the trace formula with any  observable microlocalized in a small enough neighborhood of $\gamma$ determines obviously its primitive period $T_\gamma$. We will assume that 
any multiple of $T_\gamma$ is isolated in the set of the periods of all the periodic trajectories on the same energy shell (let us remark that in case this condition is not fulfilled, 
our results remains valid by taking observables microlocalized in a neighborhood of the non-degenerate elliptic $\gamma$).
Moreover it is known (\cite{fried,gu}) that the coefficients of  trace formula determines the Poincaré angles modulo $2\pi\bbZ$ and we prove in Appendix \ref{deuxpi} 
that, in the case where $\gamma$ is not reduced to one point, any realization  of the Poincaré angles as real numbers leads to a different Birkhoff normal form but give an explicit symplectomorphism that conjugates one to another: hence, our reconstruction of the ``true" Hamiltonian is  independent of the choice of the realization. We also show that in the ``bottom of a well" case, the $\theta_i$s are determined by the spectrum.}

\tcb{
Therefore the only knowledge we will require will be the fact that there exists a geometric periodic trajectory $\gamma$, possibly of dimension zero, which is elliptic non-degenerate (see definition below) and whose set of periods in isolated in the set of periods of the same energy shell.}

We will be concerned by three cases:
\begin{enumerate}
\item\label{1} $\gamma$ is a curve
\item\label{2} the general ``bottom of a well" case ($\gamma$  reduced to a point)
\item\label{3} the ``bottom of a well" case when the Hamiltonian is a Schr\"odinger operator.
\end{enumerate}
Our results will also be of three different kinds :
\begin{itemize}
\item[a.]\label{a} The knowledge of the coefficients of the trace formula for \eqref{1}, or of some of the diagonal matrix elements (expectation values) between eigenvectors of the Hamiltonian for \eqref{2},\eqref{3}, for a  family of observables satisfying some algebraic properties on $\gamma$ determine some Fermi coordinates (see definition below). 
It is the content of Theorems \ref{mainFermi}, \ref{fermicorbot}, \ref{hope2fermi}.
\item[b.] The knowledge of the coefficients of the trace formula for \eqref{1}, or of some of the diagonal matrix elements (expectation values) between eigenvectors of the Hamiltonian for \eqref{2}-\eqref{3}, for another  family of observables, expressed on any (not necessarily the one determined by a.) Fermi system of coordinates, determine the full Taylor expansion of the total symbol of the Hamiltonian expressed on these Fermi coordinates (Theorems \ref{main}, \ref{corbot}, \ref{hope2}).
\item[c.] The combination of the two preceding cases, where the family of observables defined in a. drives the knowledge of the full Hamiltonian. More precisely, the knowledge of the quantities expressed in a. determines a family of observables, which is precisely the one defined in b. expressed in the Fermi system determined in a., the trace coefficients or some of the diagonal matrix elements of which determine the full Taylor expansion of the Hamiltonian on a determined system of coordinates (Corollaries \ref{toutper}, \ref{toutfond}, \ref{toutsch}).
\end{itemize}

Finally we obtain analog classical results as byproduct of the quantum ones in Section \ref{class}.

\begin{definition}\label{nondeg}
A periodic trajectory of the Hamiltonian flow generated by $H(x, \xi)$ is said to be non-degenerate elliptic if  its linearized Poincar\'e map has eigenvalues $(e^{\pm i\theta_i})_{1\leq i\leq n}$, $\theta_j\in\R$, and the rotation angles $\theta_i$ ($1\leq i\leq n$) and $\pi$ are independent over $\mathbb Q$.
\end{definition}

\begin{definition}[Fermi coordinates]\label{def:fermicoord}
We will denote by ``Fermi coordinates" any system of local coordinates of $T^*\mathcal M$  near $\gamma$, $(x,t,\xi,\tau)\in T^*(\bbR^n\times\bbS^1)$, such that $\gamma=\{x=\xi=\tau=0\}$ and on which the principal symbol $H_p$ of $H(x,\hbar D_x)$ can be written \tcb{for any chosen realization of the Poincaré angles $\theta_i\in\bbR$} as: 

\be
H_p (x,t,\xi,\tau)=H_0(x,t,\xi,\tau)+H_2
\ee
where 
\be\label{H_0}
H_0(x,t,\xi,\tau)=E+\sum_{i=1}^n\theta_i\frac{x_i^2+\xi_i^2}{2}+\tau
\ee
and 
\be\label{annulordre3}
H_2=O\left((x^2+\xi^2+\vert\tau\vert)^{\frac{3}{2}}\right)
\ee
\end{definition}

The existence of such local coordinates, guaranteed by the Weinstein tubular neighborhood Theorem (\cite{wei}), was proved in \cite{gu,vgtp,ze1}
 under the hypothesis of non-degeneracy mentioned earlier. However the construction of Fermi coordinates involves the knowledge of \tcr{the quadratic part of} $H_p$ in a neighborhood of $\gamma$. Our first result shows that a system of Fermi coordinates can be determined by $\gamma$ only at the classical level and some quantum spectral quantities. 
(\tcr{constructed out of a system of local coordinates near $\gamma$ and some quantum spectral quantities.}) 
 
 \begin{theorem}\label{mainFermi}

Let $\w^k_p, k=0,1,\dots ,2n^2+n, p\in\bbZ,$  be any pseudodifferential operators whose respective principal symbols $\mathcal\w^k_p$ satisfy in a local symplectic system of coordinates $(x,t,\xi,\tau)\in  T^*(\R^n\times \bbS^1)$ such that $\gamma=\{x=\xi=\tau=0\}$:
\begin{equation}\mathcal\w^0_p(x,t,\xi,\tau)=e^{-2i\pi pt}\tau \ \text{ and }\ \mathcal\w^k_p(x,t,\xi,\tau)=e^{-2i\pi pt}\uv^k(x,\xi),\ k=1,\dots, 2n^2+n\end{equation}
 with the property that 
\textbf{$\uv^k(0)=\nabla\uv^k(0)=0$ and the  Hessians  $d^2\uv^k(0)$ are linearly independent}.

An example of such symbols is given by the family,
\be\label{defobq}
\left\{\begin{array}{rcl}
\mathcal Q^1_{ijp}(x,t,\xi,\tau)&=&e^{-2i\pi pt}x_i\xi_j\\
\mathcal Q^2_{ijp}(x,t,\xi,\tau)&=&e^{-2i\pi pt}x_ix_j\\
\mathcal Q^3_{ijp}(x,t,\xi,\tau)&=&e^{-2i\pi pt}\xi_i\xi_j\\
\calQ_p(x,t,\xi,\tau)&=&e^{-2i\pi pt}\tau
\end{array}\right.
\ee
Then, the knowledge of the coefficients $(a_1^l(\w_p^k))_{0\leq k\leq 2n^2+n}$ in \eqref{gutzo}-\eqref{Sbojo}  
determines (in a constructive way) an explicit  system of Fermi coordinates near $\gamma$.

 \end{theorem}

\begin{theorem}\label{main}

 Let $\gamma$ be a non-degenerate elliptic periodic trajectory of the Hamiltonian flow generated by the principal symbol $H_p$ of  $H(x,\hbar D_x)$ on the energy shell $H_p^{-1}(E)$, and
let $(x,t,\xi,\tau)\in  \R^n\times \bbS^1\times\bbR^{n+1}
$ be a system of Fermi coordinates near $\gamma$.

\tcrr{For $(m,n,p)\in\bbN^{2n}\times\bbZ$, let  $O_{mnp}$, $O_p$ be  any pseudodifferential operator whose total Weyl  symbols (in this system of coordinates) $\mathcal{O}_{mnp}$, $\calO_p$ satisfy in a neighborhood of $\gamma$:  }
\begin{equation}\label{defo}
\left\{\begin{array}{rcl}
\mathcal{O}_{mnp}(x,t,\xi,\tau)&=&e^{-i2\pi pt}\prod\limits_{j=1}^n\left(\frac{x_j+i\xi_j}{\sqrt{2}}\right)^{m_j}\left(\frac{x_j-i\xi_j}{\sqrt{2}}\right)^{n_j}\\
&&+\sum\limits_{2l+N=|m|+|n|+1}O(\hbar^l(x^2+\xi^2+\vert\tau\vert)^{\frac{N}{2}})\\
\mathcal{O}_{p}(x,t,\xi,\tau)&=&e^{-i2\pi pt}\tau+\sum\limits_{2l+N=3}O(\hbar^l\left(x^2+\xi^2+\vert\tau\vert\right)^{\frac{N}{2}})
\end{array}\right.
\end{equation}

\tcrr{Then the knowledge of the coefficients $a^{l}_k(O_{mnp})$ and $a^{l}_k(O_{q})$  in \eqref{gutzo}-\eqref{Sbojo} for $k\leq N$ and  $m,n,p,q$ satisfying
\begin{enumerate}
\item $\vert m\vert+\vert n\vert\leq N$ \label{c1}
\item  $\forall j=1\dots n,\ m_j=0$\ \textbf{or} $ n_j =0$ \label{c2}
\item $p\in\bbZ$, $q\in\bbZ^*$ \label{c3}
\end{enumerate}}
determines the Taylor expansion near $\gamma$  up to order $M_1$ in $(x,\xi)$ and $M_2$ in $\tau$, of the total Weyl symbol,  in this  system of Fermi coordinates, of $H(x,\hbar D_x)$ up to order $l$ in $\hbar$ 
at the condition that $2l+M_1+2M_2\leq N$.
\end{theorem}
 
Concatenating the two preceding results we get the coordinate free statement:
\begin{corollary}\label{toutper}

Let $ \w_p^k$ be  as in Theorem \ref{mainFermi}. Then the knowledge of the coefficients 
$a_1^l(Q^k_{ijp})$, $a_1^l(Q_p)$   for $ p\in\bbZ$, $l\in\bbZ, 1\leq i,j\leq n, k\in\{1,2,3\}$ determine  observables $O_{mnp}$, $O_q$
out of which the coefficients $a^{l}_k(O_{mnp})$ and $a^{l}_k(O_{q})$, for $k\leq N$ and  $m,n,p,q$ satisfying conditions $(1),(2),(3)$ in Theorem \ref{main},
 determine 
 modulo a function vanishing to infinite order on $\gamma$, the full symbol of $H(x,\hbar D_x)$, in a determined system of local coordinates near $\gamma$.

\end{corollary}
\begin{remark}
It is easy to see that Condition $2$ implies that the number of observables in the transverse to $\gamma$ directions (for each Fourier coefficient in $t$) needed for determining $H(x,\hbar D_x)$ up to order $N$ is a polynomial function of $N$ of degree $n$, 
 while the number of all polynomial functions in $(x,\xi,\tau)$ of order $N$ is a polynomial in $N$ of higher degree $2n+1$. The fact that not all observables are needed can be understood by the fact that we know that the Hamiltonian we are looking for is  conjugated to the normal form by a unitary operator and not by any operator (see the discussion after Theorem \ref{flatmain}). At the classical level this is a trace of the fact that we are looking for a symplectomorphism, and not any diffeomorphism (see section \ref{class}).
\end{remark}
\begin{remark}
The asymptotic expansion of the trace \eqref{gutzo} involves only the microlocalization of $H(x,\hbar D_x)$ in a formal neighborhood of $\gamma$. Therefore there is no hope to recover from spectral data more precise information that the Taylor expansion of its symbol near $\gamma$. The rest of the symbol concerns spectral data of order $\hbar^\infty$.
\end{remark}

Let us now consider the case where $\gamma$ is reduced to one point, namely the  ``bottom of a well" case. 
Let us assume that the principal symbol $H_p$ of $H(x,\hbar D_x)$ has a global minimum at $z_0\in T^*\mathcal M$, and let $d^2H_p(z_0)$ be the Hessian of $H$ at $z_0$. Let us define the matrix $\Omega$ defined by  $d^2H_p(z_0)(\cdot, \cdot)=:\omega_{z_0}(\cdot, \Omega^{-1}\cdot)$ where $\omega_{z_0}(\cdot,\cdot)$ is the canonical symplectic form of $T^*\mathcal M$ at $z_0$. The eigenvalues of $\Omega$ are purely imaginary, let us denote them by $\pm i\theta_j$ with $\theta_j>0, \ j=1\dots n$. Let us assume moreover that $\theta_j, j=1\dots n$ are rationally independent.

 \begin{definition}
By extension of definition \ref{def:fermicoord}, we will also denote by Fermi coordinates any system of Darboux coordinates $(x,\xi)\in T^*\bbR^n$ centered at $z_0$  such that: 
 \begin{equation} H_p(x,\xi)=H_p(z_0)+\sum_{i=1}^n \theta_i \frac{x_i^2 +\xi_i^2}{2}+O((x,\xi)^3).
 \end{equation}
\end{definition}

The existence of such local coordinates will be proved in section \ref{bot}, once again by using the knowledge of \tcr{the quadratic part of} $H_p$ near $z_0$. Our next result shows  that one can explicitly construct Fermi coordinates out of the knowledge of some quantum spectral quantities. 

\begin{theorem}\label{fermicorbot}
Let $\w^k, k=1\dots 2n^2+n$ be \textbf{any} pseudodifferential operators whose principal symbols $ \mathcal\w^k$ is such that \textbf{$\mathcal\w^k(z_0)=\nabla\mathcal\w^k(z_0)=0$ and the  Hessians  $d^2\mathcal\w^k(z_0)$ are linearly independent}.

An example of such symbols is given by the family,  $1\leq i,j\leq n$, $k\in\{1,2,3\}$,
\be\label{defobq}
\left\{\begin{array}{rcl}
\mathcal Q^1_{ij}(x,\xi)&=&x_i\xi_j\\
\mathcal Q^2_{ij}(x,\xi)&=&x_ix_j\\
\mathcal Q^3_{ij}(x,\xi)&=&\xi_i\xi_j
\end{array}\right.
\ee
in any system $(x,\xi)\in T^*\R^n$ of Darboux coordinates centered at $z_0$.

Then, for any $\epsilon=\epsilon(\hbar)>0,\ \hbar=o(\epsilon(\hbar))$ (e.g.\ $\epsilon=\hbar^{1-\eta},\eta>0$), the knowledge of the spectrum of $H(x,\hbar D_x)$ in $[H_p(z_0),H_p(z_0)+\epsilon]$ 
and the diagonal matrix elements of $\w^k$ 
between the 
corresponding 
eigenvectors of $H(x,\hbar D_x)$
determines (in a constructive way) an explicit  system of Fermi coordinates.

 \end{theorem}

\begin{theorem}\label{corbot}

For $(m,n)\in\bbN^{2n}$, let   $O_{mn}$ be any pseudodifferential operator whose total Weyl symbol $\mathcal O_{mn}$ satisfy in a neighborhood of  $z_0$:
\be\label{defob}\begin{split}
\mathcal{O}_{mn}(x,\xi)=&\prod_{j=1}^n\left(\frac{x_j+i\xi_j}{\sqrt{2}}\right)^{m_j}\left(\frac{x_j-i\xi_j}{\sqrt{2}}\right)^{n_j}
+\sum\limits_{\substack{2l+N=\\|m|+|n|+1}}O\left(\hbar^l\left(x^2+\xi^2\right)^{\frac{N}{2}}\right)
\end{split}
\ee
in a system $(x,\xi)\in T^*\R^n$ of Fermi coordinates centered at $z_0$.

Then  the knowledge of the spectrum of $H(x,\hbar D_x)$ in $[H_p(z_0),H_p(z_0)+\epsilon]$ with $\hbar^{1-\alpha}=O(\epsilon)$ for some $\alpha>0$, and the diagonal matrix elements of 
$O_{mn}$ 
between the 
corresponding 
eigenvectors of $H(x,\hbar D_x)$,

 for:

\begin{enumerate}
\item $\vert m\vert+\vert n\vert\leq N$
\item  $\forall j=1\dots n,\ m_j=0$\ \textbf{or} $ n_j =0$, 
\end{enumerate}
determines the Taylor expansion up to order $N$ of the full symbol  of $H(x,\hbar D_x)$ at $z_0$  in the coordinates $(x,\xi)$.
\end{theorem}

\begin{corollary}\label{toutfond}
The diagonal matrix elements of the operators $ \w^k$ as in Theorem \ref{fermicorbot} determine 
 observables $O_{mn}$ whose diagonal matrix elements as in Theorem  \ref{corbot} determine, modulo a function vanishing to infinite order at $z_0$, the full symbol of $H(x,\hbar D_x)$, in a determined system of local coordinates near $z_0$.
\end{corollary}
\begin{remark} Although we will not prove it here, let us remark that Theorem \ref{corbot} (and also Theorem \ref{main}) is also valid in the framework of quantization of Kälherian  manifolds. 
\end{remark}

\vskip 0.5cm
In the case where $H(x,\hbar D_x)$ is a Schr\"odinger operator $-\hbar^2\Delta+V$, it is known, \cite{gur}, that the (actually classical) normal form determines  the Taylor expansion of the potential in the case where the latter is invariant, for each $i=1\dots n$, by the symmetry $x_i\to -x_i$. The same result holds without the symmetry assumption in the case $n=1$, with assumption $V'''(0)\neq 0$, as it has been shown in \cite{gcdv}.

Let now $H=-\hbar^2\Delta+V$ be a Schrödinger operator 
and $q_0$ be a global non-degenerate minimum of $V$. Let us assume that the square-roots $(\theta_i)_{1\leq i\leq n}$ of the eigenvalues of $d^2V(q_0)$
 are linearly independent over the rationals. In that precise case, we will denote by Fermi coordinates any system of Darboux coordinates $(x,\xi)\in T^*\bbR^n$, in which the (principal or total, both notions are equivalent here) symbol $H$ of our Schrödinger operator can be written as: 
 \begin{equation} 
 H(x,\xi)=V(q_0)+\sum_{i=1}^n \theta_i \frac{x_i^2 +\xi_i^2}{2}+R(x)
 \end{equation}
 where $R(x)=O(x^3)$. 
 The existence of such local coordinates will also be proved in section \ref{bot}, and Theorem \ref{hope2fermi} below proves that one can explicitly construct Fermi coordinates out of any system of local coordinates centered at $q_0$.

Theorem \ref{hope2} shows that the matrix elements of only a \textbf{finite number} of observables are necessary to recover the full Taylor expansion of the potential in the general case.

\begin{theorem}\label{hope2fermi}

Let $\w^k, k=1\dots\frac{n(n+1)}2$ be \textbf{any} pseudodifferential operators whose principal symbols are potentials $\mathcal\w^k$ such that
\textbf{$\mathcal\w^k(q_0)=\nabla\mathcal\w^k(q_0)=0$ and the  Hessians  $d^2\mathcal\w^k(q_0)$ are linearly independent}
(an example of such potentials is the family $\mathcal Q^2_{ij}(x)=x_ix_j$ in a 
local system of coordinates centered at $q_0$).

Then, for any $\epsilon=\epsilon(\hbar)>0,\hbar=o(\epsilon)$, the knowledge of the spectrum of $H(x,\hbar D_x)$ in $[V(q_0),V(q_0)+\epsilon]$ and the diagonal matrix elements  of $\w^k, k=1\dots\frac{n^2+n}2$ between the corresponding eigenvectors of $H(x,\hbar D_x)$ determines (in a constructive way) an explicit system of Fermi coordinates.

\end{theorem}
\vskip 0.5cm
\begin{theorem}\label{hope2}
 Let $(x,\xi)\in T^*\R^n$ be a system of Fermi coordinates centered at $(q_0,0)$. 
 
Then the knowledge of the spectrum of $H(x,\hbar D_x)$ in $[V(q_0),V(q_0)+\epsilon]$ with $\hbar^{1-\alpha}=O(\epsilon)$ for some $\alpha>0$, and the diagonal matrix elements  of the $2^n-1$ observables 
$O_{m0},\ m=(m_1,\dots,m_n)\in\{0,1\}^n\setminus\{0\}$, defined in Theorem \ref{corbot},
between the corresponding eigenvectors of $H(x,\hbar D_x)$
determines the full Taylor expansion  of $V$ at $q_0$ in the coordinates $x$.

\end{theorem}
\vskip 0.5cm
\tcr{
\begin{corollary}\label{toutsch}
The diagonal matrix elements of the operators $\w^k$ 
as in Theorems \ref{hope2fermi} determine $2^n-1$ observables $O_{m0}$ whose diagonal matrix elements as in Theorem  \ref{hope2} determine the potential $V$ up to a function vanishing to infinite order at $q_0$.
\end{corollary}}
\begin{remark}
Note that since we are dealing with observables localized near the bottoms of the wells, the hypothesis that $z_0$ in Theorems \ref{fermicorbot}-\ref{corbot} and $q_0$ in Theorems \ref{hope2fermi} and \ref{hope2} are \emph{global} minima can be released and the corresponding results can be formulated in a straightforward way.
\end{remark}

\vskip 0.5cm

The proof of Theorem \ref{main}  relies on two  results having their own interest per se: 
\textbf{Proposition \ref{etape1}} which shows that the coefficients of the trace formula determine the matrix elements $\langle\vp_j,O(x,\hbar D_x)\vp_j\rangle$ where $\vp_j$ are the eigenvectors of the normal form of the Hamiltonian, and \textbf{Proposition \ref{propp}} which states that the knowledge of the matrix elements of the conjugation of a given known selfadjoint operator by a 
unitary one determines, in a certain sense, the latter.

As a byproduct of Proposition \ref{propp} we obtain also a purely classical result, somehow analog of it: the averages on Birkhoff angles associated to Birkhoff coordinates of the same classical observables than the ones in Theorem \ref{main} determine the Taylor expansion of the (true) Hamiltonian. This is the content of 
 \textbf{Theorem \ref{classmain}} below.
\vskip 1cm
The paper is organized as follows. 
Section \ref{proof} is devoted to the proof of Theorems \ref{main}, \ref{corbot} and \ref{hope2}. In section \ref{bot}, we give an explicit construction of some Fermi coordinates out of any system of local coordinates in both the periodic and ``Bottom of the well" case: this is the content of Theorems \ref{mainFermi}, \ref{fermicorbot} and \ref{hope2fermi}. In Section \ref{class} we show the classical equivalent of our quantum formulation.

Through the whole paper, $\llbracket l,m\rrbracket,\ l< m,$
will stand for the set of integers $\{l,\dots,m\}$ \tcr{and we  we will   assume, without loss of generality, that the period of $\gamma$ is  equal to $1$}.
 \\

\section{Recovering the Hamiltonian in some given Fermi coordinates}\label{proof}

Let us start this section by observing that, 
by microlocalization  near $\gamma$,
it is enough, in order to prove Theorem \ref{main}, to prove  Theorem \ref{flatmain} below, which is nothing but the  same statement  expressed in a local Fermi system of coordinates.

The proof of Theorem \ref{flatmain} will  need  a construction of the quantum Birkhoff normal form, given in subsection \ref{subsec:QBNF}. The rest of the proof is then a consequence of Proposition \ref{etape1} (subsection \ref{subsec:etape1}) and Proposition \ref{propp} (subsection \ref{subsec:propp}). Subsection \ref{subsec:BW} contains the proof of the analogs of Theorem \ref{main} when $\gamma$ is reduced to a single point, both in the general and ``Schrödinger" cases (Theorems  \ref{corbot} and \ref{hope2}).

\begin{theorem}\label{flatmain}
Let $H(x,\hbar D_x)$ be a self-adjoint semiclassical elliptic pseudodifferential operator on $L^2(\R^n\times \bbS^1)$.
Let $(x,t,\xi,\tau)\in T^*(\R^n\times \bbS^1)$ be the canonical symplectic coordinates and let us  assume that $\gamma=\bbS^1=\{x=\xi=\tau=0\}$ is a  non degenerate elliptic periodic orbit of the Hamiltonian flow generated by the principal symbol $H_p$ of  $H(x,\hbar D_x)$ on the energy shell $H_p^{-1}(E)$.

Let us assume moreover that $H_p$  can be written in these coordinates as: 
\be
H_p (x,t,\xi,\tau)=H_0(x,t,\xi,\tau)+H_2
\ee

where \be\label{annulordre3}
H_2=O\left((x^2+\xi^2+\vert\tau\vert)^{\frac{3}{2}}\right)
\ee

And $H_0$ is equal to:  
\be\label{H_0}
H_0(x,t,\xi,\tau)=E+\sum_{i=1}^n\theta_i\frac{x_i^2+\xi_i^2}{2}+\tau
\ee

For $(m,n,p)\in\bbN^{2n}\times\bbZ\times$, let  $O_{mnp}$, $O_p$ be  any pseudodifferential operator whose total Weyl symbols $\mathcal{O}_{mnp}$, $\calO_p$ satisfy in a neighborhood of $\gamma$:  
\begin{equation}\label{defo}
\left\{\begin{array}{rcl}
\mathcal{O}_{mnp}(x,t,\xi,\tau)&=&e^{-i2\pi pt}\prod_{j=1}^n\left(\frac{x_j+i\xi_j}{\sqrt{2}}\right)^{m_j}\left(\frac{x_j-i\xi_j}{\sqrt{2}}\right)^{n_j}
\\&&+\sum\limits_{2l+N=|m|+|n|+1}O(\hbar^l\left(x^2+\xi^2+\vert\tau\vert\right)^{\frac{N}{2}})\\
\mathcal{O}_{p}(x,t,\xi,\tau)&=&e^{-i2\pi pt}\tau+\sum\limits_{2l+N=3}O(\hbar^l\left(x^2+\xi^2+\vert\tau\vert\right)^{\frac{N}{2}})
\end{array}\right.
\end{equation}

Then the knowledge of the coefficients $a^{l}_k(O_{mnp})$ and $a^{l}_k(O_{q})$  in \eqref{gutzo}-\eqref{Sbojo} for $k\leq N$ and  $m,n,p,q$ satisfying
\begin{enumerate}
\item $\vert m\vert+\vert n\vert\leq N$
\item  $\forall j=1\dots n,\ m_j=0$\ \textbf{or} $ n_j =0$ 
\item $p\in\bbZ$, $q\in\bbZ^*$
\end{enumerate}
determines the Taylor expansion near $\gamma$ of the full symbol (in the system of coordinates $(x,t,\xi,\tau)$) of $H(x,\hbar D_x)$ up to order $N$.
\end{theorem}

The proof of Theorem \ref{flatmain} will be  divided into  three steps: first, we will prove in Proposition \ref{W}  the existence of the quantum Birkhoff normal form in a form  convenient for our computations, especially concerning the discussion of orders. 
In Proposition \ref{etape1}, we will show that the trace formula with any observable $O$ determines the matrix elements of 
$O$ in the eigenbasis of the normal form.
Finally, in Proposition \ref{propp}, we will show that these matrix elements 
determines $H(x,\hbar D_x)$ in a formal neighborhood of $x=\xi=\tau=0$, which will lead to Theorem \ref{flatmain}. \\

Let us first fix some notations and standard results.
For $i\in\llbracket1,n\rrbracket:=\{1,\dots,n\}$, we define  the following operators on $L^2(\bbR^n\times \bbS^1)$:
\begin{itemize}
\item $a_i=\frac 1 {\sqrt 2}(x_i+\hbar\partial_{x_i})$
\item $a_i^*=\frac 1 {\sqrt 2}(x_i-\hbar\partial_{x_i})$
\item $D_t=-i\hbar\partial_t$
\item $P_i:=\frac 12 \left(-\hbar\partial^2_{x_i}+x_i^2\right)=a_i^*a_i+\frac{\hbar}{2}$
\end{itemize}
For $\mu\in\bbN^n,\ \nu\in\bbZ$ we will denote by $\mur$ the common eigenvectors of  $P_1\dots P_n$ and $D_t$:
\be\label{pi}
P_i\mur=(\mu_i+\frac 12 )\hbar\mur \mbox{ and }D_t\mur=2\pi\hbar\nu \mur.
\ee 
These vectors are explicitly constructed as follows: 

\be\label{mur}
  |0,0\rangle(x,t):= \frac{1}{(\pi\hbar)^{\frac{n}{4}}}e^{\frac{-x^2}{2\hbar}},\ \ \ 
  \mur(x,t):=e^{i2\pi\nu t}\prod_{i=1}^n\frac{1}{\sqrt{\mu_i !\hbar^{\vert\mu\vert}}}a_i^{*\mu_i}|0,0\rangle(x,t)
\ee
We will also need the notation
\be\label{mu}
\vert\mu\rangle(x):=\vert\mu,0\rangle(x,0)
\ee
We will not need the explicit expressions of $\mur(x,t)$ and $\vert\mu\rangle(x)$ in terms of rescaled Hermite functions, but rather use the following identities: 
\be\label{prop}
  \begin{cases}
   a_i\mur=\sqrt{\mu_i\hbar}|\mu_1,\dots,\mu_{i-1},\mu_{i}-1,\mu_{i+1},\dots,\mu_n,\nu\rangle \\
a_i^*\mur=\sqrt{(\mu_i+1)\hbar}|\mu_1,\dots,\mu_{i-1},\mu_{i}+1,\mu_{i+1},\dots,\mu_n,\nu\rangle \\
[a_i,a_j^{*}]=\delta_{ij}\hbar,\ 
[a_i,a_j]=0.
 \end{cases}
\ee

We shall write $\vert\mu\vert:=\sum\limits_{i=1}^n\mu_i$,  $z_i=\frac{x_i+i\xi_i}{\sqrt{2}}$, $p_i=\frac{x_i^2+\xi_i^2}{2}$ and denote by
$\text{Op}^{W}(f)$  the pseudodifferential operator whose  total Weyl symbol is $f$. We have
\be\label{autres}
\text{Op}^{W}(z_i)=a_i,\ 
\text{Op}^{W}(\bz_i)=a^{*}_i,\ 
\text{Op}^{W}(z_i\bz_i)=P_i\mbox{ and }
\text{Op}^{W}(\tau)=D_t
\ee

Finally, we will denote by $a$, $a^*$ or $P$ the $n$-tuple of  operators $a_i$, $a_i^*$, $P_i$, $i\in\llbracket1,n\rrbracket$
and denote for $j$  $n$-tuple of nonnegative integers, $X^j=\prod\limits_{i=1}^n X_i^{j_i}$.
\ \\
\subsection{Construction of the Quantum Birkhoff normal form}\label{subsec:QBNF}

Our construction of the normal form, inspired by \cite{vgtp}, is the content of the following Proposition.

\begin{proposition}\label{W}
Let $H(x,\hbar D_x)$ be a self-adjoint semiclassical elliptic pseudodifferential operator on $L^2(\bbR^n\times\mathbb{S}^1)$, whose principal symbol is 

\be
H_p (x,t,\xi,\tau)=H_0(p,\tau)+H_2
\ee
where $H_0(p,\tau)=\sum_{i=1}^n\theta_ip_i+\tau$ and $H_2$ vanishes to the third order on $x=\xi=\tau=0$.\\

Then for any $N\geq3$, there exists  a self-adjoint semiclassical elliptic pseudodifferential operator $\widetilde{W}_{\leq N}$ and a smooth function $h(p_1,\dots,p_n,\tau,\hbar)$ satisfying microlocally in a neighborhood of $x=\xi=\tau=0$ the following statement:

\be\label{BNF}\begin{split}
&\forall M>0, \ \exists C_N=C_N(M)>0, \forall(\mu,\nu,\hbar)\in\bbN^n\times \bbZ\times [0,1[, |\mu\hbar|+|\nu\hbar|<M, \\ 
 &\normleft \left(\U{N}H\Uinv{N}-h(P_1,\dots,P_n,D_t,\hbar)\right)\mur\normright\leq C_N( |\mu\hbar|+|\nu\hbar|)^{\frac{N+1}{2}}
 \end{split}
\ee

The operators $\widetilde{W}_{\leq N}$ can be computed recursively in the form: 

\be
\widetilde{W}_{\leq N}=W_{\leq N}+(D_t^2+\sum_{i=1}^nP_i)^{N+1}
\ee
where
\be
\begin{cases}
W_{\leq N}=\sum_{3\leq q\leq N} W_q\\
W_q:=\sum\limits_{2p+|j|+|k|+2m=q} \alpha_{pjkm}(t)\hbar^p\text{Op}^{W}(z^j\bz^k)D_t^m
\end{cases}
\ee
with
 $\alpha_{pjkm}$ smooth and 
 $W_q$ is symmetric. 
\end{proposition}

\begin{remark}[\textsc{Important convention}]
We are only interested in recovering the Hamiltonian in a formal neighborhood of $\gamma$: every asymptotic expansion is meant microlocally and we will be rewriting equations such as \eqref{BNF} simply as: 
\be
 \normleft \left(\U{N}H\Uinv{N}-h(P_1,\dots,P_n,D_t,\hbar)\right)\mur\normright=O\left( |\mu\hbar|+|\nu\hbar|)^{\frac{N+1}{2}}\right)
\ee
By abuse of notation, we will identify the same way any operator with its version microlocalized near $\gamma$.
\end{remark}

\begin{remark}
We introduce $\W{N}$ in order to gain ellipticity and self-adjointness like it has been done in Lemma 4.5 of \cite{vgtp}.
\end{remark}
The proof of Proposition \ref{W} will need several preliminaries: 
\begin{definition}\label{opog}

We will say that a pseudodifferential operator $A$ on $L^2(\bbR^n\times \bbS^1)$ is ``polynomial of order $r\in\bbN$" (PO($r$)) if there exists $\alpha_{pjkm}\in C^\infty(\mathbb{S}^1,\bbC)$ such that:
\be\label{PO}
A=\sum_{2p+|j|+|k|+2m=r} \alpha_{pjkm}(t)\hbar^p\text{Op}^W(z^j\bz^k) D_t^m
\ee

\end{definition}

These operators have the following  properties.

\begin{proposition}\label{propreduc}
Let $A$ be a pseudodifferential operator on $L^2( \bbR ^n\times \bbS^1)$.  
Then, there exists a family of operators $A_r$, $r\in\bbN$  such that for any $i\in\bbN$, $A_r$ is PO($r$)
and 
\be\label{reducti}
\forall N\in\bbN, \left\|\left(A-\sum_{r=0}^NA_r\right)\mur\right\|=O\left(\left(|\mu\hbar|+|\nu\hbar|\right)^{\frac{N+1}2}\right)
\ee
\end{proposition}
Let us define a notion of suitable asymptotic equivalence.
\begin{definition}\label{sim}
Let us introduce for any operator $A$ the notations $\lfloor A\rfloor_r$ et $\lfloor A\rfloor_{\leq N}$ which represent respectively the terms of order $r$ and of order smaller or equal to $N$ in the expansion \eqref{reducti}.\\
If $A$ and $B$ are two operators, we will write  $A\sim B$ if, for any $r\in\bbN$, $\lfloor A\rfloor_r=\lfloor B\rfloor_r$.\\
Also, if $(A_n)_{n\in\bbN}$ is a family of operators, we will write that: 
\be
A\sim \sum_{n=0}^{+\infty}A_n
\ee
if for any $N\in\bbN$, $\lfloor A_n\rfloor_{\leq N}$ is zero for $n$ sufficiently large, and the finite sum: 
\be
 \sum_{n=0}^{+\infty}\lfloor A_n\rfloor_{\leq N}=\lfloor A\rfloor_{\leq N}.
\ee
\end{definition}

\begin{proof}[Proof of proposition \ref{propreduc}]

 Let $a$ be the total Weyl symbol of $A$. Let us define the family $(\alpha_{pjkm})_{(p,m,j,k)\in\bbN^2\times(\bbN^n)^2}$ of functions on $\bbS^1$ by the Taylor expansion of $a$ near $z=\bz=\tau=\hbar=0$, for any $N\in\bbN$:

\begin{equation}\label{def:alpha}
a(z,t,\bz,\tau,\hbar)=\sum_{r=0}^N\sum\limits_{\substack{2p+|j|+|k|\\+2m=r}}\alpha_{pjkm}(t)\hbar^pz^j\bz^k\tau^m +\sum_{p=0}^{\frac{N+1}{2}}O\left(\hbar^p(|z|^2+|\tau|)^{\frac{N+1}{2}-p}\right)
\end{equation}
For any $r\in\bbN$, let us notice $(z,t,\bz,\tau,\hbar)\mapsto\sum_{\scriptstyle 2p+|j|+|k|+2m=r}\alpha_{pjkm}(t)\hbar^pz^j\bz^k\tau^m$ is the total symbol of a pseudodifferential operator $A_r$, which is PO($r$). And by \eqref{pi}, \eqref{autres} and \eqref{def:alpha} (see \cite{vgtp}): 
\be\begin{split}
\forall N\in\bbN, \ \left\|\left(A-\sum_{r=0}^NA_r\right)\mur\right\|&=\sum_{p=0}^{\frac{N+1}{2}}\hbar^pO\left(\left(|\mu\hbar|+|\nu\hbar|\right)^{\frac{N+1}2-p}\right)\\
&=O\left(\left(|\mu\hbar|+|\nu\hbar|\right)^{\frac{N+1}2}\right).
\end{split}
\ee
This concludes the proof. 
\end{proof}

The following lemma will be crucial for our computations.
\begin{lemma}\label{general}
Let  $F$ and $G$ be PO($r$) and PO($r'$) respectively then $\frac{[F,G]}{i\hbar}$ is PO($r+r'-2$).
\end{lemma}
\begin{proof}
The proof  of Lemma \ref{general} will be a direct consequence of the two following lemmas, whose proof will be given at the end of this proof.
\begin{lemma}\label{monotoPO}
Any monomial operator of order $r$, that is of the form $\alpha(t)\hbar^pb_1\dots b_lD_t^m$, where: 
\begin{itemize} \item for $j\in\llbracket1,l\rrbracket$, $b_j\in\{a_1,a_1^{*},\dots,a_n,a_n^*\}$ 
\item $2p+ l+2m=r$ 
\end{itemize} is PO($r$). 
\end{lemma}
\begin{lemma}\label{crochetmono}
If $F$ and $G$ are monomials of order $r$ and $r'$ respectively, then $\frac{[F,G]}{i\hbar}$ is PO($r+r'-2$)
\end{lemma}
Indeed, any PO($r$) is a finite sum of monomials of the same order, hence if $F$ and $G$ are PO($r$) and PO($r'$) respectively, then $\frac{[F,G]}{i\hbar}$ is a finite sum of quantities of type $\frac{[\widetilde{F},\widetilde{G}]}{i\hbar}$ where $\widetilde{F}$ and $\widetilde{G}$ are monomials of order $r$ and $r'$ respectively. Any of those quantities are PO($r+r'-2$) by Lemmas \ref{monotoPO} and \ref{crochetmono}, and a finite sum of PO($r+r'-2$) is PO($r+r'-2$). Lemma \ref{general} is proved. 
\end{proof}
Let us  prove now Lemmas \ref{monotoPO} and \ref{crochetmono}: 
\begin{proof}[Proof of Lemma \ref{monotoPO}]
Since for any $i,j\in\llbracket1,n\rrbracket$, $i\neq j$, $a_i$ and $a_i^*$ commute with both $a_j$ and $a_j^*$, it is sufficient to prove that any ordered product $b_1\dots b_l$, where $l\geq 1$ and for any $j\in\llbracket 1,l\rrbracket$, $b_j\in\{a_1,a_1^{*}\}$, is PO($r$). For any such ordered product, let us introduce the integer $k(b_1\dots b_l)=\sharp\{m\in\llbracket 1, l\rrbracket, b_m=a^*_1\}$. 

We will proceed by induction on $l$. Let us define for any positive integer $l$ the following assertion 

(A$_l$): 
``Any ordered product $b_1\dots b_l$, where for any $j\in\llbracket 1, l\rrbracket$,  $b_j\in\{a_1,a_1^{*}\}$, is the sum of the operator $\text{Op}^W(z_1^{l-k}\bz_1^k)$ (where $k=k(b_1\dots b_l)$ and of a linear combination of the operators $\hbar^p\text{Op}^W(z_1^j\bz_1^m)$ with $p\geq 1$, $2p+j+m=l$ and $j-m=l-2k$. ".

 If $l=1$, there is nothing to prove since $a_1=\text{Op}^W(z_1)$ and $a_1^*=\text{Op}^W(\bz_1)$.

If $l=2$, \begin{equation*}\begin{cases}a_1^2=\text{Op}^W(z_1^2)\\a_1^{*2}=\text{Op}^W(\bz_1^2)\\a_1a_1^*=P_1+\frac \hbar 2=\text{Op}^W(z_1\bz_1)+\frac\hbar 2\\ a_1^*a_1=\text{Op}^W(z_1\bz_1)-\frac{\hbar}2\end{cases}\end{equation*}
and therefore, the assertion is proved for $l=2$.

 Now, let $l$ be a positive integer, and let us assume (A$_k$) up to order $k=l$. Let $B=b_1\dots b_{l+1}$ be an ordered product, where for any $j\in\llbracket1, l+1\rrbracket$, $b_j\in\{a_1,a_1^{*}\}$.\\
If for any $j\in\llbracket1,l\rrbracket$, $b_j=b_{j+1}$, then $B=\text{Op}^W(z_1^{l+1})$ or $B=\text{Op}^W(\bz_1^{l+1})$. \\
Otherwise, one can assume that $b_1=a_1$, and that $j_0=\max\{j\in\llbracket 1,l+1\rrbracket,b_j=a_1\}$ satisfies: $1\leq j_0\leq l$. Then, we have: $[a_1^{j_0},a_1^*]=j_0\hbar a_1^{j_0-1}$, so that:
\be
b_1\dots b_{l+1}=a_1^{j_0}a_1^*b_{j_0+2}\dots b_{l+1}=a^*_1a_1^{j_0}b_{j_0+2}\dots b_{l+1}+\hbar j_0 a_1^{j_0-1}b_{j_0+2}\dots b_{l+1}
\ee
Therefore, if one sets $k:=k(b_1\dots b_{l+1})$, since $\binom{l+1}{k}=\binom{l}{k}+\binom{l}{k-1}$:
\be\label{eq:decoup}\begin{split}
\binom{l+1}{k}b_1\dots b_{l+1}=&\binom{l}{k}a_1^{j_0}a_1^*b_{j_0+2}\dots b_{l+1}+\binom{l}{k-1}a^*_1a_1^{j_0}b_{j_0+2}\dots b_{l+1}\\&+\hbar\binom{l}{k-1}j_0 a_1^{j_0-1}b_{j_0+2}\dots b_{l+1}
\end{split}
\ee

(A$_{l-1}$) gives us that $a_1^{j_0-1}b_{j_0+2}\dots b_{l+1}$ is a linear combination of the operators $\hbar^{p}\text{Op}^W(z_1^j\bz_1^m)$ with $2p+j+m=l-1$ and $j-m=l+1-2k$.

Let us now observe that 
\be
\binom{l+1}{k}Op^{W}(z^{l+1-k}\bz^k)=\binom{l}{k}a_1Op^{W}(z^{l-k}\bz^k)+\binom{l}{k-1}a^*_1Op^{W}(z^{l+1-k}\bz^{k-1})
\ee
so that  (A$_l$), for ordered products $a_{1}^{j_0-1}a_1^*b_{j_0+2}\dots b_{l+1}$ and $a_1^{j_0}b_{j_0+2}\dots b_{l+1}$, gives us, by equation that $\binom{l+1}{k}b_1\dots b_{l+1}$ is the sum of $\binom{l}{k}a_1Op^{W}(z^{l-k}\bz^k)+\binom{l}{k-1}a^*_1Op^{W}(z^{l+1-k}\bz^{k-1})=\binom{l+1}{k}Op^{W}(z^{l+1-k}\bz^k)$ and a linear combination of the operators $\hbar^p\text{Op}^W(z_1^j\bz_1^m)$ with $p\geq 1$, $2p+j+m=l+1$ and $j-m=l+1-2k$
\end{proof}

\begin{proof}[Proof of Lemma \ref{crochetmono}]
It is  sufficient  to remark that if $F$ and $G$ are of the form: 
\be\nonumber
F=\alpha(t)b_1\dots b_lD_t^m \text{ and } G=\beta(t)b'_1\dots b'_{l'}D_t^{m'}
\ee
where:
\begin{itemize}
\item $\alpha$ and $\beta$ are smooth 
\item $l+2m=r$, $l'+2m'=r'$
\item For $j\in\llbracket 1,l\rrbracket$, for $j'\in 1,l'\rrbracket$, $b_j,b'_{j'}\in\{a_1,a_1^{*}\}$
\end{itemize}
 then $\frac{[F,G]}{i\hbar}$ is a finite sum of monomials of order $r+r'-2$ since, by Lemma \ref{monotoPO}, each of them is PO($r+r'-2$).
With those assumptions on $F$ and $G$, we get: 
\be\begin{split}\label{decompcrochet1}
\frac{ [F,G]} {i\hbar}=&\frac{[\alpha(t)b_1\dots b_lD_t^m,\beta(t)b'_1\dots b'_{l'}D_t^{m'}]}{i\hbar}\\=&\alpha(t)\beta(t)\frac{[b_1\dots b_l,b'_1\dots b'_{l'}]}{i\hbar}D_t^{m+m'}+\alpha(t)b_1\dots b_l\frac{[D_t^{m},\beta(t)]}{i\hbar}b'_1\dots b'_{l'}D_t^{m'}\\-&\beta(t)b'_1\dots b'_{l'}\frac{[D_t^{m'},\alpha(t)]}{i\hbar}b_1\dots b_lD_t^m
\end{split}\ee
Therefore it is sufficient to prove that$\frac{[b_1\dots b_l,b'_1\dots b'_{l'}]}{i\hbar}$, $\frac{[D_t^{m},\beta(t)]}{i\hbar}$ and $\frac{[D_t^{m'},\alpha(t)]}{i\hbar}$ are respectively: PO($l+l'-2$), PO($2m-2$) and PO($2m'-2$) (with the convention that a PO($j$) with $j<0$ is $0$).\\
For the two last, it is quite obvious, since: 
\be
\frac{[D_t^{m},\beta(t)]}{i\hbar}=\sum_{k=0}^{m-1}\binom{m}{k}(i\hbar)^{m-k-1}\beta^{(m-k)}(t)D_t^k
\ee
Now, for $j\in\llbracket1,l'\rrbracket$, let us set $\epsilon_j=1$ if $b'_j=a^*_1$,  $\epsilon_j=-1$ otherwise. Since $[a_1,a_1^*]=\hbar$, we get: 
\be\begin{split}\nonumber
b_1\dots b_lb'_1\dots b'_{l'}=b'_1b_1\dots b_l b'_2\dots b'_{l'}&+\frac{\epsilon_1+1}{2}\hbar\sum_{\substack{k=1\\b_k=a_1}}^lb_1\dots b_{k-1}b_{k+1}\dots b_lb'_2\dots b'_{l'}\\
&+\frac{\epsilon_1-1}{2}\hbar\sum_{\substack{j=1\\b_k=a^*_1}}^lb_1\dots b_{k-1}b_{k+1}\dots b_lb'_2\dots b'_{l'}
\end{split}\ee
Hence by induction on $j\in\llbracket1,l'\rrbracket$: 
\be\begin{split}\label{decompcrochet2}
\frac{[b_1\dots b_l,b'_1\dots b'_{l'}]}{i\hbar}=&-i\sum_{j=1}^{l'}\frac{\epsilon_j+1}{2}\sum_{\substack{k=1\\b_k=a_1}}^lb'_1\dots b'_{j-1}b_1\dots b_{k-1}b_{k+1}\dots b_lb'_{j+1}\dots b'_{l'}\\
&-i\sum_{j=1}^{l'}\frac{\epsilon_j-1}{2}\sum_{\substack{k=1\\b_k=a^{*}_1}}^lb'_1\dots b'_{j-1}b_1\dots b_{k-1}b_{k+1}\dots b_lb'_{j+1}\dots b'_{l'}
\end{split}
\ee
The right-hand side of \eqref{decompcrochet2} is a finite sum of monomials of order $l+l'-2$, hence it is PO($l+l'-2$) by Lemma \ref{monotoPO}, and Lemma \ref{crochetmono} is proved.
\end{proof}

\begin{proposition}\label{vgtp}
Let $G$ be PO(r).
There exists $F$,  PO($r$), and $G_1=G_1(P_1,\dots,P_n,D_t,\hbar)$ such that: 
\be
\frac{[H_0(P, D_t),F]}{i\hbar}=G+G_1
\ee
Moreover,  $F$ is  symmetric if $G$ is symmetric, $G_1=0$  if $r$ is odd, and $G_1$ is an homogeneous polynomial function of total order $\frac{r}{2}$  if $r$ is even.
\end{proposition}

\begin{remark}\label{nonnul}
If $F=\sum_{2p+|j|+|k|+2m=r} \alpha_{pjkm}(t)\hbar^p\text{Op}^W(z^j\bz^k) D_t^m$, one can choose: 
\be
\int_{\mathbb{S}^1}\alpha_{pjjm}(t)dt=0      
\ee
Indeed, any $\text{Op}^W(z^j\bz^j)D_t^m$ commutes with $H_0(P,D_t,\hbar)$. It is the choice we will make through this article. 
\end{remark}

\begin{proof}[Proof of Proposition \ref{vgtp}]
Let us first assume that $G$ is a monomial of order $r$: $G=\beta(t)b_1\dots b_lD^m_t$ where: 
\begin{itemize}
\item $\alpha$ is smooth 
\item $l+2m=r$
\item For $j\in\llbracket1,l\rrbracket$,  $b_j\in\{a_1,a_1^{*},\dots,a_n,a_n^{*}\}$
\end{itemize}
and let us look for $F$ of the form: $F=\alpha(t)b_1\dots b_lD^m_t$.
We have:
\be\begin{split}
\frac{ [H_0,F]} {i\hbar}=&\frac{[H_0,\alpha(t)b_1\dots b_{l}D_t^{m}]}{i\hbar}\\
=&\alpha(t)\sum_{s=1}^n\theta_s\frac{[P_s,b_1\dots b_{l}]}{i\hbar}D_t^{m}+\frac{[D_t,\alpha(t)]}{i\hbar}b_1\dots b_{l}D_t^{m}\\
=&\alpha(t)\sum_{s=1}^n\theta_s\frac{[P_s,b_1\dots b_{l}]}{i\hbar}D_t^{m}+\alpha'(t)b_1\dots b_l D_t^m 
\end{split}\ee
If we set, for $s\in\llbracket 1,n\rrbracket$, $k_s=\sharp\{m\in\llbracket 1, l\rrbracket, b_m=a^{*}_{s}\}$ and $j_s=\sharp\{m\in\llbracket 1, l\rrbracket, b_m=a_s\}$, we deduce from \eqref{decompcrochet2} that: 
\be
\frac{[P_s,b_1\dots b_l]}{i\hbar}=i(j_s-k_s)b_1\dots b_l
\ee
Hence: 
\be
\frac{ [H_0,F]} {i\hbar}=i\sum_{s=1}^n\theta_s(j_s-k_s)\alpha(t)b_1\dots b_{l}D_t^{m}+\alpha'(t)b_1\dots b_l D_t^m 
\ee
$\frac{ [H_0,F]} {i\hbar}=G$ admits a solution if there exists $\alpha$ such that:
\be
i\sum_{s=1}^n\theta_s(j_s-k_s)\alpha(t)+\alpha'(t)=\beta(t)
\ee 
If $(c_p(\alpha))_{p\in\bbZ}$ and $(c_p(\beta))_{p\in\bbZ}$ are the Fourier coefficients of $\alpha$ and $\beta$, it is sufficient that, for $p\in\bbZ$, $c_p(\alpha)$ is solution of: 
\be\label{condition1alpha}
i\left(\sum_{s=1}^n\theta_s(j_s-k_s)+2\pi p\right)c_p(\alpha)=c_p(\beta)
\ee
and 
\be\label{condition2alpha}
c_p(\alpha)\underset{p\rightarrow+\infty}{=}O\left(\frac{1}{|p|^{\infty}}\right)
\ee
If the $n$-tuples $j$ and $k$ are different, the non-degeneracy condition on the $\theta_i$s together with the fact that $c_p(\beta)\underset{p\rightarrow+\infty}{=}O\left(\frac{1}{|p|^{\infty}}\right)$ (because $\beta$ is smooth), gives the existence of $c_p(\alpha)$ satisfying \eqref{condition1alpha} and \eqref{condition2alpha}.\\
If $r$ is odd, $j$ and $k$ can't be equal, hence Proposition \ref{vgtp} is proved in this case ($r$ odd and $G$ monomial)\\
If $r$ is even, and $j=k$, there exists a family $(c_p(\alpha))_{p\in\bbZ^*}$ satisfying \eqref{condition1alpha} and \eqref{condition2alpha}. Hence, if $\alpha$ is the smooth function with Fourier coefficients $c_p(\alpha)$ for $p\neq0$ and $c_0(\alpha)=0$, we get: 
\be
\frac{ [H_0,F]} {i\hbar}=G+c_0(\beta)b_1\dots b_l D_t^m
\ee
And from the proof of Lemma \ref{monotoPO}, we know that $c_0(\beta)b_1\dots b_l D_t^m$ can be reordered as the sum: $G_1(P,D_t,\hbar):=c_0(\beta)\sum_{2p+2|k|=l}a_{p,k}\hbar^p P^kD_t^m$. Therefore, Proposition \ref{vgtp} is proved in the case where $r$ is even and $G$ is monomial.\\
The general case is easily deduced from the case where $G$ is monomial, since $G$ is a finite sum of monomials of the same order. \\
Also, the form of $F$ allows us to conclude immediately that $F$ is symmetric if $G$ is so.
\end{proof}
Now we have everything we need for the proof by induction of Proposition \ref{W}. 

\begin{proof}[Proof of Proposition \ref{W}] 

Microlocally near $x=\xi=\tau=0$,  $H(x,\hbar D_x)$ satisfies, in the sense of Definition \ref{sim},
 \be
H:=\opH\sim H_0(P_1,\dots,P_n,\hbar D_t)+\sum_{q\geq 3} H_q
,\ \ \ 
H_{q}:=\lfloor H(x,\hbar D_x)\rfloor_{q}
\ee

Let us set $W_{\leq 2}=0$, and construct by induction $(W_q)_{q\geq 3}$ and $(H^q)_{q\geq 3}$, such that: 
\begin{itemize}
\item for $q\geq 3$, $W_q$ is PO($q$) and $H^{q}$ is zero if $q$ is odd, an homogeneous polynomial function of total order $\frac{q}{2}$ if $q$ is even. 

 \item and for any $q\geq 3$: 
\be \nonumber\frac{i}{\hbar}[W_{q},H_0]+H_{q}+\left\lfloor\frac{i}{\hbar}[W_{\leq q-1},H-H_0]+\sum_{l\geq2}\frac{i^l}{\hbar^l l!} [\overbrace{W_{\leq q-1},\dots,W_{\leq q-1}}^{l \ \text{times}},H]\right\rfloor_{q}=H^{q}(P,D_t,\hbar)
 \ee
\end{itemize}
The existence of such a family is guaranteed by Proposition \ref{vgtp}. 
 
 Let us set, for any $N\geq 3$, $\W{N}:=\sum_{q=3}^NW_q+(|D_t|^2+\sum_{i=1}^nP_i)^{\frac{N+1}{2}}$.
 As for any $q\geq 2$ $H^{2q}$ is an homogeneous polynomial function of total order $q$, we can choose, by Borel's lemma, a smooth function $h$ such that, for any $N\geq 2$ and in a neighborhood of $p=\tau=0$: 
\be
\left\vert h(p,\tau,\hbar)-H_0(p,\tau)-\sum_{q=2}^{N} H^{2q}(p,\tau,\hbar)\right\vert=O\left((\vert p\vert+\vert\tau\vert+\vert\hbar\vert)^{N+1}\right)
\ee
We have, for any $N\geq 3$:
 \be\nonumber\begin{split}
  \U{N} H\Uinv{N}&\sim H+\frac{i}{\hbar}[\W{N},H]+\sum_{l\geq2}\frac{i^l}{\hbar^l l!} [\overbrace{\W{N},\dots,\W{N}}^{l \ \text{times}},H]\\
 &\sim H+\frac{i}{\hbar}[W_{\leq N},H_0]+\frac{i}{\hbar}[W_{\leq N},H-H_0]+\sum_{l\geq2}\frac{i^l}{\hbar^l l!} [\overbrace{\W{N},\dots,\W{N}}^{l \ \text{times}},H]\\
 &+\frac{i}{\hbar}[\W{N}-W_{\leq N},H]
  \end{split}\ee
Since the for any $q\leq N$, $W_{q}$ is PO($q$) and $H_0$ is PO($2$), Lemma \ref{general} gives us that: 
\be
\left\lfloor\frac{i}{\hbar}[W_{\leq N},H_0]\right\rfloor_q=\frac{i}{\hbar}[W_{q},H_0]
\ee
Since the expansion of $H-H_0$ in PO($r$) contains no term of order less or equal to $2$, Lemma \ref{general} also gives  for $q\leq N$:  
 \be
\left\lfloor\frac{i}{\hbar}[W_{\leq N},H-H_0]\right\rfloor_q=\lfloor\frac{i}{\hbar}[W_{\leq q-1},H-H_0]\rfloor_q
\ee
Lemma \ref{general} finally gives us, that since the expansion of $\opH$ in PO($r$) contains no term of order less or equal to $1$ and the one of $\W{N}$ no term of order less or equal to $2$ for $q\leq N$:
\be
\left\lfloor\sum_{l\geq2}\frac{i^l}{\hbar^l l!} [\overbrace{\W{N},\dots,\W{N}}^{l \ \text{times}},H]\right\rfloor_q=\left\lfloor\sum_{l\geq2}\frac{i^l}{\hbar^l l!} [\overbrace{W_{\leq q-1},\dots,W_{\leq q-1}}^{l \ \text{times}},H]\right\rfloor_q
\ee
and since the one of $\W{N}-W_{\leq N}$ contains no term of order less or equal to $N+1$:
\be
\lfloor\frac{i}{\hbar}[\W{N}-W_{\leq N},H]\rfloor_q=0
\ee
Therefore for any $q\leq N$: 
\be
\left\lfloor\U{N} H\Uinv{N}\right\rfloor_q=H^{q}(P,D_t,\hbar)=\left\lfloor h(P,D_t,\hbar)\right\rfloor_q
\ee
Finally Proposition \ref{propreduc} gives us: 
\be
\normleft \left(\U{N}H\Uinv{N}-h(P,D_t,\hbar)\right)\mur\normright=O\left( |\mu\hbar|+|\nu\hbar|)^{\frac{N+1}{2}}\right)
\ee
which concludes the proof. 
\end{proof}
 \subsection{Recovering the matrix elements from the Trace formula}\label{subsec:etape1}
  The next result is the first inverse result needed for the proof of Theorem \ref{flatmain}. 
 \begin{proposition}\label{etape1}
Let $O$ be a pseudodifferential operator whose principal symbol vanishes on $\gamma$. 

\begin{enumerate}

\item\label{firstpoint} There exists a smooth function $f$ vanishing at $(0,0,0)$ such that for any $N\geq3$: 

  \be
\mul \U{N} O\Uinv{N}\mur=f\left((\mu+\frac{1}{2})\hbar,2\pi\nu\hbar,\hbar\right)+O\left((\vert\mu\hbar\vert+\vert\nu\hbar\vert)^{\frac{N}{2}}\right)
  \ee
  
 \noindent  Moreover let, for any integer $l$, $\phi_l$ be a Schwartz function whose Fourier transform is compactly supported in $(l-1,l+1)$ and let $(a_j^l(O))_{l\geq0}$ provided by the trace formula \eqref{Sbojo}. Then

\item\label{secondpoint} The Taylor expansion of $f$ up to order $N$  is entirely determined by the family $(a_j^l(O))$, $0\leq j\leq N$, $l\in\bbN$.

\end{enumerate}
\end{proposition}

\begin{proof}
Let us first prove point \eqref{firstpoint}.
Let us consider a monomial $G=\alpha(t)b_1\dots b_lD^m_t$ where: 
\begin{itemize}
\item $\alpha$ is smooth 
\item $l+2m=r$
\item For $j\in\llbracket1,l\rrbracket$,  $b_j\in\{a_1,a_1^{*},\dots,a_n,a_n^{*}\}$
\end{itemize}
Let us set for $i\in\llbracket1,n\rrbracket$, $k_i=\sharp\{m\in\llbracket 1, l\rrbracket, b_m=a^*_i\}$\and $j_i=\sharp\{m\in\llbracket 1, l\rrbracket, b_m=a_i\}$.

If $j\neq k$ or $\alpha\notin\bbC$, then: $\mul G\mur=0$ for any $(\mu,\nu)\in\bbN^n\times\bbZ$.

If now $j=k$ and $\alpha\in\bbC$, then there exists complex numbers $\alpha_l$ ($0\leq l_i\leq j_i$ for $i\in\llbracket1,n\rrbracket$), such that: 
\be
G=\sum_{0\leq l_i\leq j_i}\alpha_l\hbar^{\vert l\vert}P_1^{j_1-l_1}\dots P_n^{j_n-l_n}D_t^m,\ \alpha_0=\alpha
\ee
Therefore for any $(\mu,\nu)\in\bbN^n\times\bbZ$: 
\be
\mul G\mur=\sum_{0\leq l_i\leq j_i}\alpha_l\hbar^{\vert l\vert}\left(\left(\mu+\frac{1}{2}\right)\hbar\right)^{j-l}(2\pi\nu\hbar)^m
\ee
Hence, if $G$ is PO($r$), then for any $(\mu,\nu)\in\bbN^n\times\bbZ$: 
\begin{itemize}
\item $\mul G\mur=0$ if $r$ is odd. 
\item If $r$ is even, there exists an homogeneous polynomial function $g$ of order $\frac{r}{2}$ such that: 
\be
\mul G\mur=g\left((\mu+\frac{1}{2})\hbar,2\pi\nu\hbar,\hbar\right)
\ee
\end{itemize}

By Proposition \ref{propreduc} and Borel's lemma, we get that that for any operator $A$ there exists a function $g$ such that for any $(\mu,\nu)\in\bbN^n\times\bbZ$: 
\be
\mul A\mur=g\left((\mu+\frac{1}{2})\hbar,2\pi\nu\hbar,\hbar\right)+O\left((\vert\mu\hbar\vert+\vert\nu\hbar\vert)^{\infty}\right)
\ee
Hence, the only point which remains to be proved, is that the function $f$ in point \eqref{firstpoint} does not depend on $N$. It is therefore sufficient to prove that for any $q\leq N-1$, 
\be\label{constantAPCR}
\left\lfloor\U{N} O\Uinv{N}\right\rfloor_q=\left\lfloor\U{q+1} O\Uinv{q+1}\right\rfloor_q
\ee
But \eqref{constantAPCR} is a direct consequence of Lemma \ref{general}. Indeed,
\be
 \U{N} O\Uinv{N}\sim O+\sum_{l\geq1}\frac{i^l}{\hbar^l l!} [\overbrace{\W{N},\dots,\W{N}}^{l \ \text{times}},O]\\
\ee
and since the principal symbol of $O$ vanishes on $\gamma$, Lemma \ref{general} gives us for any $l\geq 1$ and any $q\leq N-1$:
\be
\left\lfloor\frac{i^l}{\hbar^l l!} [\overbrace{\W{N},\dots,\W{N}}^{l \ \text{times}},O]\right\rfloor_q=\left\lfloor\frac{i^l}{\hbar^l l!} [\overbrace{\W{q+1},\dots,\W{q+1}}^{l \ \text{times}},O]\right\rfloor_q
\ee

Let us now move on to the proof of point \eqref{secondpoint}.

Since $\hat{\phi_l}$ is supported near a single period of the flow, one can 
 microlocalize the trace formula with observables near $\gamma$:

\be\label{tf}
    \text{Tr}\left(O\phi_l\left(\frac{H-E}{\hbar}\right)\right)=\text{Tr}\left(O\int_\bbR\hat{\phi_l}(t)\rho(P_1+\dots+P_n+|\zeta|)e^{it\frac{H-E}{\hbar}}dt\right)+O(\hbar^{\infty})
\ee
where $\rho\in C^\infty_0(\bbR)$ is compactly supported and $\rho=1$ in a neighborhood of $p=\tau=0$.
Therefore we can conjugate \eqref{tf} by the microlocally unitary operator  $\U{N}$:

\be\nonumber\begin{split}
   & \text{Tr}\left(O\phi_l\left(\frac{H-E}{\hbar}\right)\right)=\\
&=\text{Tr}\left((\U{N} O\Uinv{N}\int_\bbR\hat{\phi_l}(t)\rho(P_1+\dots+P_n+|\zeta|)e^{it\frac{\U{N}H\Uinv{N}-E}{\hbar}}dt\right)+O(\hbar^{\infty})
\end{split}
\ee
Thanks to Proposition \ref{W}, we can lighten the r.h.s. for any $(\mu,\nu)\in\mathbb{N}^n\times\bbZ$
\be
  \begin{split}
  &\int_\bbR\hat{\phi_l}(t)\rho(P_1+\dots+P_n+|\zeta|)e^{it\frac{\U{N}H\Uinv{N}-E}{\hbar}}dt\mur\\
 =&\left(\int_\bbR\hat{\phi_l}(t)\rho\left((|\mu|+\frac{n}{2}+|2\pi\nu|)\hbar\right)e^{it\frac{h((\mu+\frac{1}{2})\hbar,\nu\hbar,\hbar)-E+O(|\mu\hbar|+|\nu\hbar|)^{\frac{N+1}{2}})}{\hbar}}dt\right)\mur
  \end{split}
\ee

As $\hat{\phi_l}$ is smooth and compactly supported, together with the non-degeneracy condition on the $\theta_i$s, we can assure that if we choose a sufficiently small support for $\rho$, we have for any $\eta>0$:

\be\nonumber
\begin{split}
&\left(\int_\bbR\hat{\phi_l}(t)\rho\left((|\mu|+\frac{n}{2}+|2\pi\nu|)\hbar\right)e^{it\frac{h((\mu+\frac{1}{2})\hbar,\nu\hbar,\hbar)-E+O(|\mu\hbar|+|\nu\hbar|)^{\frac{N+1}{2}})}{\hbar}}dt\right)\mur\\
&=\left(\int_\bbR\hat{\phi_l}(t)\rho\left((|\mu|+\frac{n}{2}+|2\pi\nu|)\hbar^{\eta}\right)e^{it\frac{h((\mu+\frac{1}{2})\hbar,\nu\hbar,\hbar)-E+O(|\mu\hbar|+|\nu\hbar|)^{\frac{N+1}{2}})}{\hbar}}dt\right)\mur+O(\hbar^{\infty})
  \end{split}
\ee
Hence, choosing $\eta<\frac{1}{2}$: 
\be\nonumber
  \begin{split}
    &\text{Tr}\left(O\phi_l\left(\frac{H-E}{\hbar}\right)\right)+O(\hbar^{\infty})\\
=&\sum_{\mu,\nu}\mul \U{N} O\Uinv{N} \mur\times\int_{\bbR}\hat{\phi_l}(t)\rho\left((|\mu|+\frac{n}{2}+|\nu|)\hbar^{\eta}\right)e^{it(2\pi\nu+\theta.(\mu+\frac{1}{2}))}\dots\\
&\dots \exp\left(\frac{it}{\hbar}\sum_{1\leq q\leq N-2}H^{q}\left((\mu+\frac{1}{2})\hbar,\nu\hbar,\hbar\right)+O\left((|\mu|+|\nu|)^{\frac{N+1}{2}}\hbar^{\frac{N-1}{2}}\right)\right)dt\\
&=\sum_{\mu,\nu}\int_{\bbR}\hat{\phi_l}(t)\rho\left((|\mu|+\frac{n}{2}+|2\pi\nu|)\hbar^{\eta}\right)e^{it(2\pi\nu+\theta.(\mu+\frac{1}{2}))}\\
&\left(1+\sum_{i\geq1}^{\frac{N-1}{2}}\hbar^iQ_i(\mu+\frac{1}{2},\nu,t)\right)\times\sum_{p\geq 1}^{\frac{N+1}{2}}\sum_{|k|+m\leq p} b_{k,m,p-|k|-m}(\mu+\frac{1}{2})^k(2\pi\nu)^m\hbar^pdt +O(\hbar^{\frac{N+1}{2}})
  \end{split}
\ee
where for any $i\leq \frac{N-1}{2}$, $Q_i$ is a determined polynomial function,  of degree in  $\left(\mu+\frac 1 2,\nu\right)$ less or equal to $i+1$, which depends on the $H^{q}$s and the Taylor expansion of $\exp$, and the 
$b_{k,m,s}$ ($(k,m,s)\in\bbN^{n+2}\backslash\{0\}$) come from the Taylor expansion at $(0,0,0)$ of the function $f$ defined in the first point of Proposition \ref{etape1}, \textit{i.e.} for any $N\geq 1$: 
\be
f(x,y,z)=\sum_{1\leq |k|+m+s\leq N}b_{k,m,s} x^k y^m z^s+O\left(|x|+|y|+|z|)^{N+1}\right)
\ee
\ \\
 Now, let us set: 
    \be
    \forall t\in\bbR^*, \forall \alpha\in (\bbR\backslash\frac{2\pi}{t}\bbZ)^n, g(t,\alpha):=\frac{e^{i\frac{t}{2}(\alpha_1+\dots+\alpha_n)}}{\prod_i(1-e^{it\alpha_i})}
    \ee
    By the non-degeneracy condition on the $\theta_i$s, $g$ is well defined on the compact support of $\hat{\phi_l}$ around a single period, which is precisely $l$. It also implies that $\theta_i.\mu$ is bounded below by $C\vert\mu\vert$ (where $C>0$) as $\vert\mu\vert$ goes to $\infty$. \\
Therefore we get from the Poisson formula and the Riemann-Lebesgue lemma that the  quantity $X_p(l)$ below can be computed recursively on $p\leq \frac{N+1}{2}$ from the $a_j^l(O)$, $j=0,\dots,p$: 

\be\label{1stpart}\begin{split}
X_p(l)&=\sum_{|k|+m\leq p}b_{k,m,p-|k|-m}\left[\left(-i\frac{\partial}{\partial t}\right)^m\left(\hat{\phi_l}(t)\left(\frac{-i}{t}\right)^k\frac{\partial^k g}{\partial\alpha^k}(t,\alpha)\right)\right](l,\theta)\\
&=\sum_{|k|+m\leq p}b_{k,m,p-|k|-m}\left[\left(-i\frac{\partial}{\partial t}\right)^m\left(-i\frac{\partial}{t\partial\alpha}\right)^kg\right](l,\theta)\\
\end{split}\ee
since $\hat{\phi_l}$ is identically $1$ around $l$.\\
Now, let us set, for any $i\in\llbracket 1,n\rrbracket$, any $t\in\bbR$ and any $\alpha\in (\R\backslash\frac{2\pi}{t}\bbZ)^n$, $x_i(t,\alpha)=e^{i\frac{t\alpha_i}{2}}$. 
and also define holomorphic function $h$ on $\bbC\backslash\{-1,1\}$ by $h(z)=\frac{z}{1-z^2}$ for $z\in\bbC\backslash\{-1,1\}$. 
We have for any $k\in\bbN^n$:  
\be
\left(-i\frac{\partial}{t\partial\alpha}\right)^kg=\prod_{i=1}^{n}\left(-i\frac{\partial}{t\partial\alpha_i}\right)^{k_i}(h\circ x_i)
\ee

For any $i\in\llbracket 1,n\rrbracket$, an easy induction on $k_i\in\bbN$ leads to the following, since for any $z\in\bbC\backslash\{-1,1\}$, $h(z)=\frac{1}{2}\left(\frac{1}{1-z}-\frac{1}{1+z}\right)$, and $-i\frac{\partial x_i}{t\partial\alpha_i}=\frac{1}{2}x_i$:
\be
\left(-i\frac{\partial}{t\partial\alpha_i}\right)^{k_i}(h\circ x_i)=\frac{k_i!}{2^{k_i+1}}\left(\frac{x_i}{(1-x_i)^{k_i+1}}+\frac{x_i}{(1+x_i)^{k_i+1}}\right)
\ee
Now, since $-i\frac{\partial x_i}{\partial t}=\frac{\alpha_i}{ 2} x_i$, an induction on $s_i\in\bbN$ shows that: 
\be\label{2ndpart}
\left(-i\frac{\partial}{\partial t}\right)^{s_i}\left(-i\frac{\partial}{t\partial\alpha_i}\right)^{k_i}(h\circ x_i)=\frac{(k_i+s_i)!\alpha_i^{s_i}}{2^{k_i+s_i+1}}\left(\frac{x_i}{(1-x_i)^{k_i+s_i+1}}+\frac{x_i}{(1+x_i)^{k_i+s_i+1}}\right)
\ee
Let us  now introduce for any n-tuple $s$ such that $|s|=m$, the multinomial coefficient: \be\nonumber\binom{m}{s}=\frac{m!}{s_1!\dots s_n !}\ee
We have: 
\be\label{3rdpart}
\left(-i\frac{\partial}{\partial t}\right)^m\left(-i\frac{\partial}{t\partial\alpha}\right)^kg=\sum_{|s|=m}\binom{m}{s}\prod_{i=1}^n\left(-i\frac{\partial}{\partial t}\right)^{s_i}\left(-i\frac{\partial}{t\partial\alpha_i}\right)^{k_i}(h\circ x_i)
\ee
Let us use Kronecker theorem, whose hypothesis is precisely the non-degeneracy condition on the $\theta_i$s: for any $n$-tuple $(x_1,\dots,x_n)\in\mathbb{S}_1^n$, one can find a sequence of integers $(l_p)_{p\in\bbZ}$, such that:
\be\nonumber
\forall j\in\llbracket 1, n\rrbracket, \ x_j(l_p,\theta)\tendvers{p}{+\infty}x_j
\ee
Therefore, setting, for any $(x_1, \dots, x_n)\in(\mathbb{S}_1\backslash\{-1,1\})^n$ and $(k,m)\in\bbN^{n+1}$:
\be\label{new}
u^{(k,m)}=\sum_{|s|=m}\binom{m}{s}\prod_{i=1}^n\frac{(k_i+s_i)!\theta_i^{s_i}}{2^{k_i+s_i+1}}\left(\frac{x_i}{(1-x_i)^{k_i+s_i+1}}+\frac{x_i}{(1+x_i)^{k_i+s_i+1}}\right)
\ee

we have that \eqref{1stpart}, \eqref{2ndpart} and \eqref{3rdpart} together with Kronecker theorem allows us to conclude that 
$
X_p:=\sum_{|k|+m\leq p}b_{k,m,p-|k|-m}u^{(k,m)}
$ is determined by the $a_j^l(O)$, $j=0,\dots,p$. 

Hence, the only thing which remains to be proved is that, if one chooses  $x_i$ tending to $1$ in a way convenient to us, the $|u^{(k,m)}|$s will be ordered for $\ll$. More precisely, let  $x_i$ tend to 1 in such a way  that  
\be\label{limitxi}
  \forall i\in\llbracket 1,n-1\rrbracket, |1-x_i|=o\left(|1-x_{i+1}|^{p}\right)
\ee
we have 
that $s_1=m$ gives the leading order in \eqref{new} and therefore

\be
(1-x_1)^m u^{(k,m)}\sim C \prod_{i=1}^n\frac{1}{(1-x_i)^{k_i+1}}
\ee
for some $C>0$. Hence, if one sets $\widetilde{m}=(m,0,\dots,0)$: 
\be
  u^{(k,m)}=o\left(u^{(k',m')}\right)\ \text{if } k+\widetilde{m}<k'+\widetilde{m'}
\ee
where $<$ is the lexicographical order on $\bbN^{n}$. Therefore, for any $p\in\bbN$ and  $(k,m)\in\bbN^{n+1}$ such that $|k_0|+m_0\leq p$, the following quantity can be recursively determined from $X_p$: 
\be\label{Xkm}
X_{k_{0},m_{0}}=\sum_{k'+\widetilde{m'}=k+\widetilde{m}}b_{k,m,p-|k|-m}u^{(k,m)}
\ee
Reversing for example the roles of $i=1$ and $i=2$ in \eqref{limitxi}, and observing that $k_2+m\neq k_2'+m'$ if $k+\widetilde{m}=k'+\widetilde{m'}$ and $(k,m)\neq (k',m')$, one determines $b_{k,m,p-|k|-m}$ from \eqref{Xkm} recursively on $m$. 
Finally, each $b_{k,m,s}$ with $|k|+m+s\leq N$ is determined by the $a_j^l(O)$, with $j=0\dots N$ and $l\in\bbN$ and the point \eqref{secondpoint} is proved, which ends the proof of Proposition \ref{etape1}.

\end{proof}

\subsection{Recovering the Hamiltonian from matrix elements}\label{subsec:propp}

In order to finish the proof of Theorem \ref{flatmain} we will show how the knowledge of the diagonal matrix elements of a given known selfadjoint operator conjugated by a 
unitary one determines the latter (in  the framework of asymptotic expansion).

Let $\widetilde{W}_{\leq N}$ as in Proposition \ref{W} and $O_{mnp},O_p $ as in Theorem \ref{flatmain}.
By Proposition \ref{etape1}, there exists smooth functions $f_{mnp}$ and $f_{p}$ vanishing at $(0,0,0)$ if $(m,n)\neq (0,0)$ such that for any $N\geq3$: 

  \be\label{etape1mnp}
\mul \U{N} O_{mnp} \Uinv{N}\mur=f_{mnp}\left((\mu+\frac{1}{2})\hbar,2\pi\nu\hbar,\hbar\right)+O\left((\vert\mu\hbar\vert+\vert\nu\hbar\vert)^{\frac{N}{2}}\right)
  \ee
and \begin{equation}\label{etape1p}
\mul \U{N} O_{p} \Uinv{N}\mur=f_{p}\left((\mu+\frac{1}{2})\hbar,2\pi\nu\hbar,\hbar\right)+O\left((\vert\mu\hbar\vert+\vert\nu\hbar\vert)^{\frac{N}{2}}\right)
\end{equation}

\begin{proposition}\label{propp} 
The Taylor expansions, at the origin, of the functions $f_{mnp}$, $f_q$  up to order $N-1,\ N\geq 3,$ for  $(m,n,p,q)\in\bbN^{2n}\times \bbZ^2$ satisfying conditions

\begin{enumerate}
\item $0<\vert m\vert+\vert n\vert\leq N$
\item  $\forall j=1\dots n,\ m_j=0$\ \textbf{or} $ n_j =0$ 
\item $p\in\bbZ$, $q\in\bbZ^*$
\end{enumerate}

determine completely  $W_{\leq N}$.
\end{proposition}


\begin{proof}[Proof of Proposition \ref{propp}]

   Let us write  
  \be\begin{split}
 W_N&=\sum\limits_{2l+|j|+|k|+2s=N} \alpha_{ljks}(t)\hbar^l\text{Op}^W(z^j\bz^{k})D_t^s\\
 &:=\sum\limits_{2l+|j|+|k|+2s=N}\sum_{d\in\bbZ} \alpha_{ljksd}\hbar^le^{i2\pi dt}\text{Op}^W(z^j\bz^{k})D_t^s
 \end{split}
 \ee
 where, for every $\alpha_{ljjs0}$ is chosen to be zero by the convention of remark \ref{nonnul}.

Since $W_2=0$  we can proceed by induction on $N\geq 3$: let's assume  $W_{\leq N-1}$ already determined. 

Let $(m,n,p,q)\in\bbN^{2n}\times \bbZ\times\bbZ^*$ be such that:  \begin{equation}\label{eq:condmn0}
0<\vert m\vert+\vert n\vert\leq N,\ \  \forall i\in\llbracket 1, n\rrbracket, \ m_in_i=0\end{equation}

 Let us also state the following lemma, whose proof will be given after the end of the present proof. 
 \begin{lemma}\label{equivalent} Let $(j,k,s,d)\in\bbN^{2n+1}\times\bbZ$, such that: $\vert j\vert +\vert k\vert +2s=N$.
 
 If $j+m=k+n$, then:
 \be\label{eq:commjks1}\begin{array}{r}
 \mul  [e^{i2\pi pt}\text{Op}^W(z^j\bz^{k})D_t^s,O_{mnp}]\mur=-\hbar g_{jks}\left(\left(\mu+\frac{1}{2}\right)\hbar,\nu\hbar\right)\\
 +O\left(\hbar^2(\vert\mu\hbar\vert+\vert\nu\hbar\vert)^{\frac{N+\vert m\vert+\vert n\vert}{2}-2}
 +\hbar(\vert\mu\hbar\vert+\vert\nu\hbar\vert)^{\frac{N+\vert m\vert+\vert n\vert-1}{2}}\right)
\end{array} \ee
 where:
  \be\nonumber
  g_{jks}\left(\left(\mu+\frac{1}{2}\right)\hbar,\nu\hbar\right)=(2\pi\nu\hbar)^{s}(\mu\hbar)^{\max(j,k)}\left(\sum_{i=1}^n\frac{k_im_i-j_in_i}{\mu_i\hbar}+\frac{ps}{\nu\hbar}\right)
  \ee
  and $\max(j,k)=(\max(j_i,k_i))_{1\leq i\leq n}$.

If $j+m\neq k+n$ or $d\neq p$, then:

\be\label{eq:commjks2}\begin{split}
 \mul  [e^{i2\pi dt}\text{Op}^W(z^j\bz^{k})D_t^s,O_{mnp}]\mur=& O\left(\hbar^2(\vert\mu\hbar\vert+\vert\nu\hbar\vert)^{\frac{N+\vert m\vert+\vert n\vert}{2}-2}\right)\\&+O\left(\hbar(\vert\mu\hbar\vert+\vert\nu\hbar\vert)^{\frac{N+\vert m\vert+\vert n\vert-1}{2}}\right)
 \end{split}
 \ee
 
We also have, if $j=k$: 

\begin{equation}\label{eq:commjks3}\begin{split}
\mul  [e^{i2\pi qt}\text{Op}^W(z^j\bz^k) D_t^s,O_{q}]\mur=-2\pi\hbar q(1+s) \left(\left(\mu+1/2\right)\hbar\right)^j(\nu\hbar)^s\\+O\left(\hbar^2(\vert\mu\hbar\vert+\vert\nu\hbar\vert)^{\frac{N -2}{2}}+\hbar(\vert\mu\hbar\vert+\vert\nu\hbar\vert)^{\frac{N+1}{2}}\right)
\end{split}
\end{equation} 
And if $j\neq k$ or $d\neq q$: 
\be\label{eq:commjks4}
 \mul  [e^{i2\pi dt}\text{Op}^W(z^j\bz^{k})D_t^s,O_{q}]\mur=O\left(\hbar^2(\vert\mu\hbar\vert+\vert\nu\hbar\vert)^{\frac{N -2}{2}}\right)+O\left(\hbar(\vert\mu\hbar\vert+\vert\nu\hbar\vert)^{\frac{N+1}{2}}\right)
 \ee
 \end{lemma}

By equation \eqref{etape1mnp}, the Taylor expansion  of function $f_{mnp}$ up to order $N-1$ determines modulo \scriptsize$O\left((|\mu\hbar|+|\nu\hbar|)^{N}\right)$\normalsize:
\be\label{W2N}
\mul \U{2N} O_{mnp}\Uinv{2N}\mur-  \mul   O_{mnp}\mur
\ee 
Since  $\W{2N}$ is a sum of polynomial operators of order greater that $3$, we get from Lemma \ref{general} that : 
\begin{equation}
\sum_{l\geq 2}\frac{i^l}{\hbar^l l!} \mul [\overbrace{\W{2N},\dots,\W{2N}}^{l \ \text{times}},O_{mnp}]\mur=O\left((|\mu\hbar|+|\nu\hbar|)^{\frac{N+\vert m\vert+\vert n\vert-1}{2}}\right)
\end{equation}
Hence, using the notations of Lemma \ref{equivalent},  \eqref{W2N} is equal,  modulo known terms and $O\left((|\mu\hbar|+|\nu\hbar|)^{\frac{N+\vert m\vert+\vert n\vert-1}{2}}\right)+O\left(\hbar(\vert\mu\hbar\vert+\vert\nu\hbar\vert)^{\frac{N+\vert m\vert+\vert n\vert}{2}-2}\right)$  to: 

 \be\label{eg2N}\sum_{\substack{|j|+|k|+2s=N+1\\j+m=k+n}} i\alpha_{0jksp} g_{jks}\left(\left(\mu+\frac{1}{2}\right)\hbar,\nu\hbar\right) 
\ee

Let us define the set $\Gamma=\{(j,k,s)\in\bbN^{2n+1}\ \vert \ |j|+|k|+2s=N,\ j+m=k+n\}$. \\
Let us choose $\mu_1(\hbar),\dots\mu_n(\hbar),\nu(\hbar)$ such that, as $\hbar$ tends to $0$:
 \be\label{ordremui1}
\nu^{\frac{N-2}{N-1}}\ll \mu_1\ll \dots\ll \mu_n\ll \nu\ll\hbar^{-\frac{1}{3}} 
 \ee 
 where $\ll$ is defined by $f\ll g\Leftrightarrow f\underset{\hbar\rightarrow 0}{=}{o}(g)$.
 
  Let us also define $i_0:=\min\{i\in\llbracket 1,n\rrbracket, m_i\neq n_i\}$ (it exists since $(m,n)\neq (0,0)$ and for any $i\in\llbracket 1,n\rrbracket$, $m_in_i=0$).
 Let us also remark $j_{i_0}n_{i_0}-k_{i_0}m_{i_0}$ never vanishes on $\Gamma$.  We have by \eqref{ordremui1} that, for $(j,k,s)\in\Gamma$,\\

 \be\label{equivg}
  g_{jks}\left(\left(\mu+\frac{1}{2}\right)\hbar,\nu\hbar\right) \underset{\hbar\rightarrow 0}{\sim}\frac{j_{i_0}n_{i_0}-k_{i_0}m_{i_0}}{\mu_{i_0}\hbar}(2\pi\nu\hbar)^{s}\prod_{i=1}^n(\mu_i\hbar)^{\max(j_i,k_i)}
  \ee 
  Let us now define a strict total order $\prec$ on $\Gamma$ by: 
  \begin{equation}\begin{array}{c}
  (j,k,s)\prec (j',k',s') \\
  \Updownarrow\\ (\max(j_1,k_1),\dots,\max(j_n,k_n),s)<(\max(j'_1,k'_1),\dots,\max(j'_n,k'_n),s')
  \end{array}
  \end{equation}
  where $<$ is the lexicographical order on $\bbN^{n+1}$. $\prec$ is asymmetric since for $i\in\llbracket1,n\rrbracket$, the sign of $m_i-n_i$ determines whether $\max(j_i,k_i)$ is equal to $j_i$ or $k_i$. 
\eqref{ordremui1} and \eqref{equivg}  give  that: 
  
   \be
   (j,k,s)\prec (j',k',s') \Rightarrow g_{jks}\left(\left(\mu+\frac{1}{2}\right)\hbar,\nu\hbar\right){\ll} g_{j'k's'}\left(\left(\mu+\frac{1}{2}\right)\hbar,\nu\hbar\right)
   \ee
and for any $(j,k,s)\in \Gamma$: \be\nonumber O\left((|\mu\hbar|+|\nu\hbar|)^{\frac{N+\vert m\vert+\vert n\vert-1}{2}}\right)+O\left(\hbar(\vert\mu\hbar\vert+\vert\nu\hbar\vert)^{\frac{N+\vert m\vert+\vert n\vert}{2}-2}\right)\ll g_{jks}\left(\left(\mu+\frac{1}{2}\right)\hbar,\nu\hbar\right)\ee
  
 Therefore, the Taylor expansion up to order $N-1$ of the functions $f_{mnp}$ determines the coefficients $(\alpha_{0jksp})_{|j|+|k|+2s=N,j+m=k+n}$  by induction on $(\Gamma,<)$.
   
Let $(m,n,p)$ run over all the possible values in $\bbN^{2n}\times\bbZ$ while satisfying condition \eqref{eq:condmn0}.  We claim that one can determine every function $\alpha_{0jks}$ with $|j|+|k|+2s= N$ and $j\neq k$.
Indeed, for any $(j,k,s)\in \bbN^{2n+1}$ such that $|j|+|k|+2s=N$ and $j\neq k$, let us choose for any $i\in\llbracket 1,n\rrbracket$:  
\be
n_i=\max(j_i-k_i,0) \text{ and } m_i=\max(k_i-j_i,0)
\ee
then $j+m=k+n$ and $(m,n)\neq(0,0)$ while for any $i\in\llbracket 1,n\rrbracket$, $m_i=0$ or $n_i=0$. Finally, \be\nonumber \vert m\vert+\vert n\vert=\sum_{i=1}^n\vert j_i-k_i|\leq |j|+|k|\leq N\ee

Let us remark that condition $j\neq k$ is always satisfied if $N$ is odd and $|j|+|k|+2s=N$. If $N$ is even, the Taylor expansion up to order $ \frac{N}{2}$ of the function $f_{q}$ determines  modulo known terms and $O\left((|\mu\hbar|+|\nu\hbar|)^{\frac{N+2}{2}}\right)+O\left(\hbar(|\mu\hbar|+|\nu\hbar|)^{\frac{N-2}{2}}\right)$: 

 \be\sum_{2|j|+2s=N} i\alpha_{0jjsq} 2\pi q(1+s) \left(\left(\mu+1/2\right)\hbar\right)^j(\nu\hbar)^s
\ee

Let us choose $\mu_1(\hbar),\dots\mu_n(\hbar),\nu(\hbar)$ such that, as $\hbar$ tends to $0$:
 \be\label{ordremui2}
\nu^{\frac{N-2}{N}}\ll \mu_1\ll \dots\ll \mu_n\ll \nu\ll\hbar^{-\frac{1}{2}} 
 \ee  
We have, for any $(j,s,q)$ such that $2\vert j\vert+2s=N$: \be\nonumber  O\left((|\mu\hbar|+|\nu\hbar|)^{\frac{N+2}{2}}\right)+O\left(\hbar(|\mu\hbar|+|\nu\hbar|)^{\frac{N-2}{2}}\right)\ll \left(\left(\mu+1/2\right)\hbar\right)^j(\nu\hbar)^s \ee
Thus, every $\alpha_{0jjsq}$ is  determined by induction on the set $\{2\vert j\vert+2s=N\}$ ordered by the lexicographical order. Hence, letting $q$ run over $\bbZ^*$, we finally determined every $\alpha_{0jksd}$ with $\vert j\vert+\vert k\vert +2s=N$ and $d\neq 0$ if $j=k$, hence the principal symbol of $W_{N}$. 

Let us now choose $1\leq l_0< \frac{N}{2}$ and assume that we already determined the functions $\alpha_{ljks}$ with $2l+|j|+|k|+2s= N$ and $l<l_0$. 
 Let $(m,n,p,q)\in\bbN^{2n}\times \bbZ\times\bbZ^*$ be such that:  \begin{equation}\label{eq:condmn2}
0<\vert m\vert+\vert n\vert\leq N-2l_0,\ \  \forall i\in\llbracket 1, n\rrbracket, \ m_in_i=0\end{equation} 
The Taylor expansion of $f_{mnp}$ up to order $N-1-l_0$ determines 
$$\sum_{\substack{2l_0+|j|+|k|+2s=N\\j+m=k+n}} i\alpha_{l_0jksp}\hbar^{l_0} g_{jks}\left(\left(\mu+\frac{1}{2}\right)\hbar,\nu\hbar\right) $$

modulo known terms $O\left((|\mu\hbar|+|\nu\hbar|)^{\frac{N+\vert m\vert+\vert n\vert-1}{2}}\right)+O\left(\hbar^{l_0+1}(|\mu\hbar|+|\nu\hbar|)^{\frac{N-2l_0+\vert m\vert+\vert n\vert}{2}-2}\right)$.

Let us choose $\mu_1(\hbar),\dots\mu_n(\hbar),\nu(\hbar)$ such that, as $\hbar$ tends to $0$,
 \be\label{ordremui3}
\nu^{\frac{N-l_0-2}{N-l_0-1}}\ll \mu_1\ll \dots\ll \mu_n\ll \nu\ll\hbar^{-\frac{1}{2l_0+3}} 
 \ee 
 
 Then for any $(j,k,s)$ such that  $2l_0+|j|+|k|+2s=N$ and $j+m=k+n$, we have: $O\left((|\mu\hbar|+|\nu\hbar|)^{\frac{N+\vert m\vert+\vert n\vert-1}{2}}\right)+O\left(\hbar^{l_0+1}(|\mu\hbar|+|\nu\hbar|)^{\frac{N-2l_0+\vert m\vert+\vert n\vert}{2}-2}\right)\ll \hbar^{l_0}g_{jks}\left(\left(\mu+\frac{1}{2}\right)\hbar,\nu\hbar\right)$. Therefore, every $\alpha_{l_0jksp}$ with $2l_0+|j|+|k|+2s=N$ and $j+m=k+n$ is determined just like before. Letting $(m,n,p)$ run over all the possible values in $\bbN^{2n}\times\bbZ$ while satisfying \eqref{eq:condmn2}, we determined every $\alpha_{l_0jksp}$ with $2l_0+|j|+|k|+2s=N-1$ and $j\neq k$. The Taylor expansion of $f_q$ up to order $\frac{N}{2}$ determines the remaining $\alpha_{l_0jjsq}$, and finally, every function $\alpha_{l_0jks}$ where $(j,k,s)$ satisfies $l_0+|j|+|k|+2s=N$, which concludes our proof by induction.\end{proof}
 \begin{proof}[Proof of Lemma \ref{equivalent}]
 The principal symbol of $\frac{1}{i\hbar}[e^{i2\pi dt}\text{Op}^W{z^j\bz^k}D_t^s,O_{mnp}]$ is: 
 
 \be\begin{split}
 \sigma_{jkds}(z,t,\bz,\tau)=&\left\{e^{i2\pi dt} z^j\bz^k\tau^s, \mathcal{O}_{mnp}\right\}\\=&\left\{e^{i2\pi dt} z^j\bz^k\tau^s,e^{-i2\pi pt} z^m\bz^n\right\}\\&+O\left((\vert z\vert^2+\tau)^{\frac{\vert m\vert+\vert n\vert+N-1}{2}}\right)
 \end{split} \ee

Hence:  \be\label{eq:calculcrochet}\begin{split}
 \sigma_{jkds}(z,t,\bz,\tau)=&-i\sum_{i=1}^n\frac{\partial}{\partial z_i}(e^{i2\pi dt} z^j\bz^k\tau^s)\frac{\partial}{\partial \bz_i}(e^{-i2\pi pt} z^m\bz^n)\\
 &+i\sum_{i=1}^n\frac{\partial}{\partial \bz_i}(e^{i2\pi dt} z^j\bz^k\tau^s)\frac{\partial}{\partial z_i}(e^{-i2\pi pt} z^m\bz^n)\\
 &-\frac{\partial}{\partial \tau}(e^{i2\pi dt} z^j\bz^k\tau^s)\frac{\partial}{\partial t}(e^{-i2\pi pt} z^m\bz^n)\\
 &+O\left((\vert z\vert^2+\tau)^{\frac{\vert m\vert+\vert n\vert+N-1}{2}}\right)\\
 =&-ie^{i2\pi(d-p)t}z^{j+m}\bz^{k+n}\tau^{s}\left(\sum_ {i=1}^n\frac{j_in_i-k_im_i}{z_i\bz_i}- \frac{2\pi ps}{\tau}\right)\\
 &+O\left((\vert z\vert^2+\tau)^{\frac{\vert m\vert+\vert n\vert+N-1}{2}}\right)
 \end{split} \ee

Let us remark that if $j+m=k+n$, then $j+m=k+n=\max(j,k)$.
Let us also remark that if $j_in_i-k_im_i\neq 0$, then $j_i+m_i\neq 0$ and $k_i+n_i\neq 0$. Therefore, the last line in \eqref{eq:calculcrochet} can be reduced to a polynomial expression modulo $O\left((\vert z\vert^2+\tau)^{\frac{\vert m\vert+\vert n\vert+N-1}{2}}\right)$. 
 $\frac{1}{\hbar}[e^{i2\pi dt}\text{Op}^W{z^j\bz^k}D_t^s,O_{mnp}]$ has the same principal symbol as the polynomial operator obtained when replacing each $z_i$ by $a_i$, $\bz_i$ by $a_i^*$, $\tau$ by $D_t$ in this polynomial expression. Since the expansion in PO $\frac{1}{\hbar}[e^{i2\pi dt}\text{Op}^W{z^j\bz^k}D_t^s,O_{mnp}]$ starts at order $N+\vert m\vert+\vert n\vert-2$, we can hence conclude that the asymptotic expansions \eqref{eq:commjks1} and \eqref{eq:commjks2} are verified. 
 
Now, the principal symbol of $\frac{1}{i\hbar}[e^{i2\pi dt}\text{Op}^W{z^j\bz^k}D_t^s,O_{q}]$ is, modulo $O\left((\vert z\vert^2+\tau)^{\frac{N+1}{2}}\right)$: 
 
 \be\begin{split}
 \tilde{\sigma}_{jkds}(z,t,\bz,\tau)=&\left\{e^{i2\pi dt} z^j\bz^k\tau^s, \mathcal{O}_{q}\right\}\\=&\left\{e^{i2\pi dt} z^j\bz^k\tau^s,e^{-i2\pi qt} \tau\right\}\\
 =&\frac{\partial}{\partial t}(e^{i2\pi dt} z^j\bz^k\tau^s)\frac{\partial}{\partial \tau}(e^{-i2\pi qt} \tau)\\&-\frac{\partial}{\partial \tau}(e^{i2\pi dt} z^j\bz^k\tau^s)\frac{\partial}{\partial t}(e^{-i2\pi qt} \tau)\\
 =&i2\pi(d+sq)e^{i2\pi (d-q)t} z^j\bz^k\tau^s
 \end{split}
  \ee 
Hence, \eqref{eq:commjks3} and \eqref{eq:commjks4} are verified just as before. 

\end{proof}

Theorem \ref{flatmain} is, as it has already been said, a direct consequence of Propositions \ref{etape1} and \ref{propp}.

\subsection{``Bottom of a well"}\label{subsec:BW} In this subsection, we treat the ``Bottom of a well" analogs of Theorem \ref{main}, namely Theorems \ref{corbot} and \ref{hope2}.
The proof of Theorem \ref{corbot} is a line by line analog of Theorem \ref{main} \tcb{after noticing that the knowledge of the spectrum near the bottom determines  the left hand side of the trace formula}: we omit it here. However, Theorem \ref{hope2}, that needs less assumptions in the particular case of a Schrödinger operator, deserves a proper proof. 

\begin{proof}[Proof of Theorem \ref{hope2}]

In a system of Fermi coordinates, the (principal and total) symbol of our Schrödinger operator can be written as: 

\begin{equation}
H(x,\xi)=V(q_0)+\sum_{i=1}^n \theta_i \frac{x_i^2 +\xi_i^2}{2}+R(x),\ \ R(x)=O(x^3)
\end{equation}

Let  $H_0(x,\xi)=\sum_{i=1}^n \theta_i \frac{x_i^2 +\xi_i^2}{2}$.
Let us state the following lemma, which is a classical analog of Proposition \ref{vgtp} (we therefore omit its proof) and uses the hypothesis of rational independence of the $\theta_i$s.

\begin{lemma}\label{lem:vgtpclass}
Let $G\in \calC^{\infty}(T^*(\bbR^n),\bbR)$ be an homogeneous polynomial of degree $k\geq 3$. There exists a unique couple of functions $G_1\in \calC^{\infty}(\bbR^n,\bbR)$ and $F\in \calC^{\infty}(T^*(\bbR^n),\bbR)$ such that 
\be
 \forall (x,\xi)\in T^*(\bbR^n), \ \{H_0,F\}(x,\xi)=G(x,\xi)-G_1(p) 
\end{equation}
and $F$ is polynomial with no diagonal term when written as a function of $(z,\bz)$ (\emph{i.e.} of the form $z^l\bz^l$)  

Moreover:  

\begin{enumerate}
\item $F$ is an homogeneous polynomial of degree $k$  and is entirely determined by the extra-diagonal terms of  $G$, \emph{i.e.} of the form $z^l\bz^m$ ($l\neq m$) with $z=(x+i\xi)/\sqrt 2$ 
\item $G_1$  is an homogeneous polynomial of degree $\frac{k}{2}$ if $k$ is even, zero otherwise. Moreover, $G_1(z\bz)$ is equal to the sum of the diagonal terms of $G$. 
\end{enumerate}
\end{lemma} 

Just like in the proof of Proposition \ref{W}, one shows recursively, using Lemma \ref{lem:vgtpclass}, the existence of a family of real numbers $(\alpha_{lm})_{l,m\in\bbN}$ 
such that if the functions $(F_N)_{N\geq 3}$ are defined for $N\geq 3$ by: 

\begin{equation}
F_N(z,\bz)=\sum_{\vert l\vert+\vert m\vert=N}\alpha_{lm} z^l\bz^m
\end{equation}
there exists homogeneous polynomials $H^i\in\calC^{\infty}(\bbR^n,\bbR)$ of degree $i$  satisfying, for $N\geq 3$: 

\begin{equation}\label{verif}
 H\circ\exp\chi_{F_{\leq N}}(x,\xi)=\sum_{i=1}^{\lfloor\frac{N}{2}\rfloor} H^{i}(p)+O((x,\xi)^{N+1})
\end{equation}
Here $p=p(x,\xi)=(\frac{x_i^2+\xi_i^2}{2})_{1\leq i\leq n}$,  $F_{\leq N}=\sum_{k=1}^N F_k$ and $\chi_{F_{\leq N}}$ is the vector field: 
\begin{equation}\label{chi}
\chi_{F_{\leq N}}=\sum_{i=1}^n\frac{\partial F_{\leq N}}{\partial\xi_i}\frac{\partial}{\partial x_i}-\frac{\partial F_{\leq N}}{\partial x_i}\frac{\partial}{\partial \xi_i}
\end{equation} 
$\sum_{i=1}^{+\infty} H^i$ (well defined modulo a flat function) is the classical Birkhoff normal form of $H$.

Let us also define for $k\in\bbN^n$, $\vert k\vert\geq 3$, $a_k=\frac{1}{k!}\frac{\partial^{\vert k\vert} R}{\partial x^k }(0)$. We observe that, for $k\in\bbN^n$: 

\begin{equation}
\begin{split}
x^{k}=\left(\frac{z+\bz}{\sqrt{2}}\right)^k
&=\frac{1}{\sqrt{2}^{\vert k\vert}}\sum_{\substack{(l,m)\in\bbN^n \\l+m=k}} \prod_{j=1}^n \binom{k_j}{m_j} \  z^l\bz^m
\end{split}
\end{equation}

Let us define $\mathcal{K}=\{k\in\bbN^n, \vert k\vert\geq 3\}\setminus 2\bbN^n$.
By lemma \ref{lem:vgtpclass}, there exists a unique homogeneous polynomial of degree $\vert k\vert\geq 3$ with no diagonal terms, such that:

\begin{equation}\label{defIk}
 \ \{H_0,I_k\}(x,\xi)=\left\{\begin{array}{ll} x^k &\text{ if } k\in\calK\\
 x^k - \frac{1}{\sqrt{2}^{\vert k\vert}}\prod_{j=1}^n \binom{k_j}{k_j/2} \  \vert z\vert^{k}&\text{ if } k\in 2\bbN^n\end{array}\right.
\end{equation}

Functions $(F_N)_{N\geq 3}$ and $(H^i)_{i\geq 1}$ are constructed recursively as follows:
let $N\geq 2$ and let us assume that we already constructed $F_3,\dots,F_N$ ($F_2=0$), and $H_1,\dots, H^{\lfloor \frac{N}{2}\rfloor}$ ($H_1(p)=\sum_{i=1}^n\theta_ip_i$). Let us set:  
\begin{equation}
  G_{N+1}(x,\xi)=H\circ\exp \chi_{F_{\leq N}}(x,\xi)-\sum_{i=1}^{\lfloor \frac{N}{2}\rfloor}H^{i}(p)+O(\norm{(x,\xi)}^{N+1})
\end{equation}
and defined $F_{N+1}$ and, if $N$ is odd, $H^{\frac{N+1}{2}}$ by Lemma \ref{lem:vgtpclass}:
\begin{equation}
 \{H_0,F_{N+1}\}(x,\xi)=\left\{\begin{array}{ll} G_{N+1}(x,\xi) &\text{ if } N \text{ is even}\\
 G_{N+1}(x,\xi)-H^{\frac{N+1}{2}}(p) &\text{ if } N \text{ is odd}\end{array}\right.
\end{equation}
We remark that, in our case, $(x,\xi)\mapsto G_{N+1}(x,\xi)-\sum_{\vert k\vert =N+1} a_kx^k$ is a sum of  terms that depend only on $F_{\leq N}$, $(H^i)_{1\leq i\leq \lfloor \frac{N}{2}\rfloor}$ and $(a_k)_{\vert k\vert\leq N}$. Therefore, we get by induction that the function(s): \begin{equation}\label{verif2}
F_{N+1}-\sum_{\vert k\vert =N+1} a_kI_k\ \text{ and when}\ N \ \text{ is odd,}\
H^{\frac{N+1}{2}}(p)-\sum_{\vert l\vert=\frac{N+1}{2}}\frac{a_{2l}}{2^{\vert l\vert}}\prod_{j=1}^n \binom{2l_j}{l_j} \   p^{l} 
\end{equation} depend only on $(a_k)_{\vert k\vert\leq N}$.

Now, let us define, for $k\in\bbN^n$, $(l_k,m_k)\in\bbN^{2n}$ by 
their components : for $i\in\llbracket 1,n\rrbracket$, $(l_k)_i=\lfloor \frac{k_i}{2}\rfloor$, $(m_k)_i=k_i-\lfloor \frac{k_i}{2}\rfloor$.
$k\mapsto (l_k,m_k)$ is a bijective correspondence between $\calK$ and the set $\Lambda$ defined by: \begin{equation}\Lambda=\{(l,m)\in\bbN^{2n}\ \vert\ m-l\in\{0,1\}^n\setminus\{0\}, \vert l\vert+\vert m\vert\geq 3\}\end{equation} 
 
 Moreover, for $k\in\calK$, $I_k$ is entirely determined by \eqref{defIk} 
and is equal  in $z,\bar z$ coordinates to:
\begin{equation}\label{eq:defIkimpair}
I_k(z,\bz)=\frac{1}{\sqrt{2}^{\vert k\vert}}\sum_{\substack{(l,m)\in\bbN^n \\l+m=k}}\frac{\prod_{j=1}^n \binom{k_j}{m_j}}{\theta.(l-m)} \  z^l\bz^m
\end{equation}

Therefore, if $k\in\calK,\ \vert k\vert=N+1$, we get by \eqref{verif2} that:
\begin{equation}\label{truc1}
\alpha_{l_km_k}-\frac{a_k}{\sqrt{2}^{\vert k\vert}}\frac{\prod_{j=1}^n \binom{k_j}{\lfloor k_j/2\rfloor}}{\theta.(l_k-m_k)}
\end{equation} 
depends only on $(a_k)_{\vert k\vert\leq N}$.

If now $k\in 2\bbN^n,\ \vert k\vert=N+1$, 
and if we write $H^{\frac{N+1}2}(p)=\sum_{\vert l\vert=\frac{N+1}{2}}b_lp^l$ we get by \eqref{verif2} that:
\begin{equation}\label{truc2}
b_{k/2}-\frac{a_{k}}{\sqrt{2}^{\vert k\vert}}\prod_{j=1}^n \binom{k_j}{k_j/2} \   
\end{equation}depends only on $(a_k)_{\vert k\vert\leq N}$. 

Therefore we get by \eqref{truc1} and \eqref{truc2} that  the family $(a_k)_{\vert k\vert=N+1}$ can be determined from the terms of order $N+1$ in the Taylor expansion  of the classical Birkhoff normal form, the family $(\alpha_{l_km_k})_{\vert k\vert =N+1}$ and the family $(a_k)_{\vert k\vert\leq N}$. 

So we just proved, by induction, that for any $N\geq 3$, $(a_k)_{\vert k\vert\leq N}$ is determined by the Taylor expansion of the classical Birkhoff normal form up to order $N$ and the family $(\alpha_{lm})_{(l,m)\in\Lambda,\vert l\vert+\vert m\vert\leq N}$.

As shown in \cite{vgtpau,gur}, the Taylor expansion of the classical Birkhoff normal form is determined by the spectrum of $H(x,\hbar D_x)$ in $[V(q_0),V(q_0)+\epsilon], \ \epsilon>0 $. 
\tcb{In fact it is obvious that one can take $\epsilon$ in the $\hbar$-dependent form given in Theorem \ref{hope2} (and Theorem \ref{corbot}) since the proof goes along the trace formula argument, and eigenvalues above this value of $\epsilon$ gives a $\hbar^\infty$ contribution to the trace formula.}

Moreover we have for $N\geq 2$ and $m\in\{0,1\}^n\setminus\{0\}$
\begin{equation*}\begin{split}
\mathcal{O}_{m0}\circ\exp\chi_{F_{\leq N+1}}(x,\xi)&=\mathcal{O}_{m0}(x,\xi)+\{F_{\leq N+1},\mathcal{O}_{m0} \}(x,\xi)+O((x,\xi)^{N+\vert m\vert })\\
&=z^m-\sum_{\substack{\\(l,k)\in\Lambda\\\vert l \vert+\vert k\vert= N+1\\ k-l=m}}\alpha_{lk}\vert z\vert^{2l}\sum_{i=1}^nk_im_i+\dots +O((x,\xi)^{N+\vert m\vert })
\end{split}
\end{equation*}
where $\dots$ stands for  extra-diagonal terms and terms which depends only on  $(\alpha_{lk})_{(l,k)\in\Lambda,\vert l\vert+\vert m\vert\leq N}$. 
Therefore, the diagonal matrix elements of an observable $O_{m0}$ 
is equal,  modulo terms  depending only on $(\alpha_{lk})_{(l,k)\in\Lambda,\vert l\vert+\vert k\vert\leq N}$, to
\begin{equation}
\sum_{\substack{\\(l,k)\in\Lambda\\\vert l \vert+\vert k\vert\leq N+1\\ k-l=m}}\alpha_{lk}\vert \mu\hbar\vert^{l}\sum_{i=1}^nk_im_i+O(\hbar) +O(\vert\mu\hbar\vert^{\frac{N+\vert m\vert }{2}})
\end{equation}

This shows, as in the proof of Theorem \ref{flatmain}, that the $\alpha_{lm},\ (l,m)\in\Lambda$, are all determined, so 
the full Taylor expansion of $R$, hence of $V$, near $q_0$, is completely determined. 
\end{proof}

\section{Explicit construction of Fermi coordinates}\label{bot}
In this section we prove Theorems \ref{mainFermi}, \ref{fermicorbot}, and \ref{hope2fermi}, 
using Lemmas  \ref{lem:diagsymp},  \ref{lem:invariantS} and \ref{lem:diagortho} on linear 
and bilinear 
algebra.
We start by the ``bottom of a well",  toy  model   for the periodic trajectory case.

\subsection{General ``Bottom of a well" case}

\begin{proof}[Proof of Theorem \ref{fermicorbot}]

 Let $(x,\xi)\in T^*(\bbR^n)$ be a system of Darboux coordinates centered at $z_0$. $d^2H_p(z_0)$ is a positive bilinear form on $T_{z_0}(T^*\calM)$, therefore, by lemma \ref{lem:diagsymp}, there exists a local change of variable $\phi$,  symplectic and linear  in the Darboux coordinates, such that: 
 \begin{equation}\label{eq:reducphi}H_p\circ\phi(x,\xi)=H_p(z_0)+\sum_{i=1}^n \theta_i \frac{x_i^2 +\xi_i^2}{2}+O(\norm{(x,\xi)}^3).
 \end{equation}

We will prove that the diagonal matrix elements of the family of pseudodifferential operators $\w^k$ in the system of eigenvectors corresponding to eigenvalues of $H(x,\hbar D_x)$ in $[H_p(z_0),H_p(z_0)+\epsilon(\hbar)]$   
provides an explicit construction of such a symplectomorphism $\phi$ (which is not unique). 

We first start with  the case where the family $(\mathcal\w^k)_{1\leq k\leq 2n^2+n}$ is realized by the example \eqref{defobq}.

Let $S$ be the matrix of $d\phi_{z_0}$ in the basis $(\frac{\partial}{\partial x_1},\frac{\partial}{\partial \xi_1}, \dots,\frac{\partial}{\partial x_n},\frac{\partial}{\partial \xi_n})$. We have for $(i,j)\in\llbracket 1,n\rrbracket^2$ and $s\in\{1,2,3\}$:%
\begin{equation}\label{verifclass}\begin{split}
 \calQ^s_{i,j}
 \circ\phi(x,\xi)&=\left(\sum_{k=1}^n S_{i^s,2k-1}x_k+S_{i^s,2k}\xi_k\right)\left(\sum_{k=1}^nS_{j^s,2k-1}x_k+S_{j^s,2k}\xi_k\right)
\\
&=\sum_{k,k'=1}^nS_{i^s,2k-1}S_{j^s,2k'-1}x_k x_{k'}+\sum_{k,k'=1}^nS_{i^s,2k}S_{j^s,2k'-1}\xi_k x_{k'}\\
&+\sum_{k,k'=1}^nS_{i^s,2k-1}S_{j^s,2k'}x_k \xi_{k'}+\sum_{k,k'=1}^nS_{i^s,2k}S_{j^s,2k}\xi_k \xi_{k'}\\
&=\sum_{k=1}^n\left[S_{i^s,2k-1}S_{j^s,2k-1}+S_{i^s,2k}S_{j^s,2k}\right]z_k \bz_{k}+R,
\end{split}
\end{equation}
where, for $(i,j)\in\llbracket 1,n\rrbracket^2$, $i^s=\left\{\begin{array}{ll}
2i-1 &\text{ if } s\in\{1,2\}\\
2i &\text{ if } s=3
\end{array}\right.$ and $j^s=\left\{\begin{array}{ll}
2j &\text{ if } s\in\{1,3\}\\
2j-1 &\text{ if } s=2
\end{array}\right.$,
and $R$ is a linear combination of  terms of the form $z_kz_{k'}$ ($(k,k')\in\llbracket 1,n\rrbracket$) and $z_k\bz_{k'}$ ($(k,k')\in\llbracket 1,n\rrbracket$, $k\neq k'$).

Let $A_\phi$ be any Fourier integral operator implementing locally  $d\phi_{z_0}$ and $\vert\mu\rangle$ defined by \eqref{mu}. The condition that $A_{\phi}^{-1}\vert\mu\rangle$ belongs to the spectral interval defined in Theorem \ref{fermicorbot} reads  as $\vert\mu\hbar\vert\leq\epsilon$. We get from \eqref{verifclass} that: 
\begin{equation}\label{verif3}
\langle\mu\vert A_{\phi}Q^s_{i,j}A_{\phi}^{-1}\vert\mu\rangle=\sum_{k=1}^n\left[S_{i^s,2k-1}S_{j^s,2k-1}+S_{i^s,2k}S_{j^s,2k}\right]\left(\mu_k+\frac{1}{2}\right)\hbar+O(\hbar)
\end{equation}
the term $O(\hbar)$ coming form the subsymbols contribution (let us recall we are microlocalized in a bounded neighborhood of $z_0$).
Therefore \eqref{verif3} for $\vert\mu\hbar\vert\leq\epsilon$ with the condition $\hbar=0(\epsilon)$ determine the values of  $S_{i,2k-1}S_{j,2k-1}+S_{i,2k}S_{j,2k}$  for $(i,j)\in\llbracket 1,2n\rrbracket^2$ and $k\in\llbracket 1,n\rrbracket$.

As  claimed by Lemma \ref{lem:invariantS}, the preceding quantities are independent of the choice of a symplectic matrix $S$ satisfying \eqref{eq:reducphi}. Since, as we already said, such a matrix $S$ is not unique, it is not possible to determine $S$ out of the preceding matrix elements. However, by Lemma \ref{lem:invariantS}, the family $(S_{i,2k-1}S_{j,2k-1}+S_{i,2k}S_{j,2k})_{(i,j)\in\llbracket 1,2n\rrbracket^2,k\in\llbracket 1,n\rrbracket}$ (determined by the preceding matrix elements) allows us to construct explicitly a suitable matrix $S$, hence a suitable symplectomorphism $\phi$.

This ends the proof in the case where the family $(\mathcal\w^k)_{1\leq k\leq 2n^2+n}$ is realized by the example \eqref{defobq}. Let us now consider the general case. The family of Hessian matrices $(d^2\mathcal\w^k(z_0))_{1\leq k\leq 2n^2+n}$, forms a basis of the space of $2n\times 2n$ symmetric matrices. Hence, each $d^2Q_{i,j}^s(z_0)$ for $(i,j)\in\llbracket 1,n\rrbracket^2$ and $s\in\{1,2,3\}$ is a linear combination of the matrices $d^2\mathcal\w^k(z_0)$, $1\leq k\leq 2n^2+n$. Since  $\mathcal\w^k(z_0)=\nabla\mathcal\w^k(z_0)=0$, there exists a family $(\lambda^k_{ijs})_{(i,j,s,k)\in\llbracket 1,n\rrbracket^2\times\{1,2,3\}\times\llbracket 1,2n^2+n\rrbracket}$ of complex numbers such that for any $(i,j,s)\in\llbracket 1,n\rrbracket^2\times\{1,2,3\}$: 
\begin{equation}
\calQ^s_{i,j}(x,\xi)=\sum_{m=1}^{2n^2+n}\lambda^k_{ijs}\mathcal\w^k(x,\xi)+O(\norm{(x,\xi)}^3)
\end{equation}
and therefore:
\begin{equation}
\langle\mu\vert A_{\phi}Q^s_{i,j}A_{\phi}^{-1}\vert\mu\rangle=\sum_{k=1}^{2n^2+n}\lambda^k_{ijs}\langle\mu\vert
A_{\phi}P^k A_{\phi}^{-1}\vert\mu\rangle+O(\hbar)+O(\vert\mu\hbar\vert^2)
\end{equation}
Hence, the family $(S_{i,2k-1}S_{j,2k-1}+S_{i,2k}S_{j,2k})_{(i,j)\in\llbracket 1,2n\rrbracket^2,k\in\llbracket 1,n\rrbracket}$ is determined just as before, and this ends the proof in the general case.
\end{proof}

\subsection{The``Schrödinger case"}

\begin{proof}[Proof of Theorem \ref{hope2fermi}] Let $x\in \bbR^n$ be \textbf{any} system of local coordinates centered at $q_0\in\calM$, and $(x,\xi)\in T^*(\bbR^n)$ the corresponding Darboux coordinates centered at $(q_0,0)\in T^*\calM$.  $d^2V(q_0)$ being a positive bilinear form on $T_{q_0}\calM$,  there exists, by Lemma \ref{lem:diagortho}, a local change of variable $u$, linear and  orthogonal in the Darboux coordinates, such that:  

\begin{equation}
V\circ u(x)=\frac{1}{2}\sum_{i=1}^n\theta_i^2 x_i^2+O(x^3)
\end{equation}
where the $\theta_i^2$s are the eigenvalues of $d^2V(q_0)$. 

Let us denote by $U$ the matrix of $du_{q_0}$ written in the basis $(\frac{\partial}{\partial x_1},\dots,\frac{\partial}{\partial x_n})$, and define a symplectomorphism $\phi$ locally by its expression in the Darboux coordinates:  $\phi(x,\xi)=(Ux,U\xi)$.

If $\phi_0$ is the symplectomorphism  sending $(x,\xi)$ to $(\frac{x_1}{\sqrt{\theta_1}},\dots,\frac{x_n}{\sqrt{\theta_n}},\sqrt{\theta_1}\xi_1,\dots,\sqrt{\theta_i}\xi_n)$, and $H$ is the (principal and total) symbol of the considered Schrödinger operator then: 
\begin{equation}\label{eq:diagHV}
H\circ\phi\circ\phi_0(x,\xi)=V(q_0)+\sum_{i=1}^n \theta_i \frac{x_i^2 +\xi_i^2}{2}+O(x^3)
\end{equation}
Just as in proof of Theorem \ref{fermicorbot}, the diagonal matrix elements of the family of the pseudodifferential operators $(Q^2_{ij})_{1\leq i,j\leq n}$ 
in the system of eigenvectors corresponding to eigenvalues of $H(x,\hbar D_x)$ in $[V(q_0),V(q_0)+\epsilon(\hbar)]$  determine the family $(U_{ik}U_{jk})_{1\leq i,j,k\leq n}$.
An orthogonal matrix $U$ such that \eqref{eq:diagHV} is verified is not unique, therefore it is not possible to determine the matrix $U$ from the preceding diagonal matrix elements. However, by Lemma \ref{lem:diagortho}, the  family $(U_{ik}U_{jk})_{1\leq i,j,k\leq n}$ does not depend on the suitable matrix $U$ (\emph{i.e.} orthogonal and satisfying \eqref{eq:diagHV}), and as we just saw it is determined by the preceding matrix elements. Therefore, one can determine the absolute values of the coefficients of any suitable matrix $U$, and also, for any $k\in\llbracket 1,n\rrbracket$, an index $i_k\in\llbracket 1,n\rrbracket$, such that $U_{i_kk}\neq 0$. The choice of the sign of $U_{i_kk}$ then determines the sign of every other coefficient of the $k$-th column. Therefore, one can determine the $2^n$ suitable matrices, corresponding to $n$ choices of signs, as claimed by Lemma \ref{lem:diagortho}. Choosing one of them determines (explicitly) a suitable symplectomorphism $\phi$.

\end{proof}

\subsection{The periodic trajectory case}

\begin{proof}[Proof of Theorem \ref{mainFermi}]
Let $X,H(x,\hbar D_x),E,\gamma$ be as in Theorem \ref{mainFermi}.
We first recall \cite{gu,vgtp,wei,ze1} that there exists a (non unique) symplectomorphism $\phi$ from a neighborhood of $\bbS^1$ in $T^*(\R^n\times \bbS^1)$ in a neighborhood of $\gamma$ in $T^{*}(X)$ such that
\begin{equation}\label{eq:H0phi} H_p\circ\phi(x,t,\xi,\tau)=H_0+H_2 \ \text{ and }\  \gamma(t)=\phi(0,t,0,0). \end{equation}

with $H_0$ and $H_2$ as in is defined as in \eqref{H_0} and \eqref{annulordre3}.
Expressing $\phi$ in a system a local coordinates $(x',\xi',t',\tau')$ near $\gamma$ such that $\gamma=\{x'=\xi'=\tau'=0\}$, one can assume that: 
\be\label{verif4} \phi(x,t,\xi,\tau)=\phi_S(x,t,\xi,\tau)=(S(t)(x,\xi),t,\tau+q_S(t,x,\xi))
\ee
Here, for any $t\in\bbS^1$,\ $S(t)$ is a linear symplectic change of variable (identified with its matrix in our system of coordinates), $q_S(t,\cdot,\cdot)$ is quadratic, $q_S(t,0,0)=0$ and
\begin{equation}\label{eq:qu}
dq_S=\left(\sum_{i=1}^{n}\dot{L}_{i+n}(t).(x,\xi) L_i(t)-\dot{L}_i(t).(x,\xi)  L_{i+n}(t) \right).(dx,d\xi)
\end{equation}
where for $i\in\llbracket 1,2n\rrbracket$ and $t\in\bbS^1$, $L_i(t)$ is the $i$-th line of the matrix $S(t)$, $\dot{L }$ the derivation with respect to $t$, and for two line vectors of size $2n$, $u.v$ is their canonical scalar product. 

For $(i,j)\in\llbracket 1,n\rrbracket^2$, $p\in\bbZ$ and $s\in\{1,2,3\}$, let $A_{S}$ be any Fourier integral operator implementing $\phi_S$. We have 
\begin{equation}
\mul A_{S} \w_p^k A_{S}^{-1}\mur=\sum_{k=1}^nc_p\left(S^{\sigma}_{i^s,2k-1}S^{\sigma}_{j^s,2k-1}+S^{\sigma}_{i^s,2k}S^{\sigma}_{j^s,2k}\right)
\left(\mu_k+\frac{1}{2}\right)\hbar+O(\hbar)
\end{equation}  
where $c_p(\cdot)$ maps a function to its $p$-th Fourier coefficient, $\sigma$ is the permutation defined by \eqref{eq:defsigma}, $S^{\sigma}$ is defined by conjugation by the permutation matrix associated to $\sigma$ just as in \eqref{eq:defSsigma}, and 
where, for $(i,j)\in\llbracket 1,n\rrbracket^2$, \\$i^s=\left\{\begin{array}{ll}
2i-1 &\text{ if } s\in\{1,2\}\\
2i &\text{ if } s=3
\end{array}\right.$, and  $j^s=\left\{\begin{array}{ll}
2j &\text{ if } s\in\{1,3\}\\
2j-1 &\text{ if } s=2
\end{array}\right.$. 

Now, just as in the proof of Proposition \ref{etape1}, the coefficients 
$\left(a_1^l(\w_p^k)\right)_{l\in\bbZ}$
determine $c_p\left(S^{\sigma}_{i^s,2k-1}S^{\sigma}_{j^s,2k-1}+S^{\sigma}_{i^s,2k}S^{\sigma}_{j^s,2k}\right)$ for any $k\in\llbracket 1,n\rrbracket$. 
Therefore, if the coefficients 
$a^l_1(\w_p^k)$
are given for any $l\in\bbZ$, $(i,j)\in\llbracket 1,n\rrbracket^2$, $p\in\bbZ$ and $s\in\{1,2,3\}$, then the functions
 \begin{equation}\label{eq:egAS} A_{i,j,k}:=S^{\sigma}_{i,2k-1}S^{\sigma}_{j,2k-1}+S^{\sigma}_{i,2k}S^{\sigma}_{j,2k}\end{equation} are determined for any $(i,j)\in\llbracket 1,2n\rrbracket^2$ and $k\in\llbracket 1,n\rrbracket$.

An easy adaptation of the proof of Lemma \ref{lem:invariantS} shows that, once the set of functions $(A_{i,j,k})_{(i,j)\in\llbracket 1,2n\rrbracket^2,k\in\llbracket 1,n\rrbracket}$ is given, one can construct explicitly a particular smooth function $\bbS^1\ni t\mapsto S_0(t)$ with values in the set of symplectic matrices, such that equality \eqref{eq:egAS} holds. We also get that any matrix $S^{\sigma}$ such that equality \eqref{eq:egAS} holds is related to $S_0$ by the equality $S^{\sigma}=S_0^{\sigma}U$ where $t\mapsto U(t)$ is a smooth function that takes its values in the set of block diagonal matrices whose diagonal blocks are $2$ by $2$ rotations.

Now let us consider this particular $S_0$ and let $U$ be any smooth function 
that takes his values in the set of block diagonal matrices whose diagonal blocks are $2$ by $2$ rotations. Let us finally define $S$ by the relation $S^{\sigma}=S_0^{\sigma}U$.
Since for any $t\in\bbS^1$, $q_{S}(t,\cdot,\cdot)$ is quadratic, we have: 
\begin{equation}\begin{split}
\mathcal\w_p^0\circ\phi_{S}(x,t,\xi,\tau)&=e^{-2i\pi pt}\tau+e^{-2i\pi pt}q_{S}(t,x,\xi)\\
&=e^{-2i\pi pt}\tau+e^{-2i\pi pt}\sum_{k=1}^n\left(\frac{\partial^2q_{S}}{\partial x_k^2}+\frac{\partial^2q_{S}}{\partial \xi_k^2}\right)(t)z_k\bz_k+R
\end{split}
\end{equation}
where $R$ is a linear combination of  terms of the form $e^{-2i\pi pt}z_kz_{k'}$ ($(k,k')\in\llbracket 1,n\rrbracket$) and $e^{-2i\pi pt}z_k\bz_{k'}$ ($(k,k')\in\llbracket 1,n\rrbracket$, $k\neq k'$).

Just as before, the coefficients 
$\left(a_1^l(\w_p^0)\right)_{l,p\in\bbZ}$ 
determine the family of  functions (of $t$ only) $\left(\frac{\partial^2q_{S}}{\partial x_k^2}+\frac{\partial^2q_{S}}{\partial \xi_k^2}\right)_{k\in\llbracket 1,n\rrbracket}$.

Now, we get from equation \eqref{eq:qu} that, for $k\in\llbracket 1,n\rrbracket$ and $t\in\bbS^1$: 
\begin{equation}\begin{split}
\frac{\partial^2q_{S}}{\partial x_k^2}(t)+\frac{\partial^2q_{S}}{\partial \xi_k^2}(t)&=\sum_{i=1}^n\dot{S}_{i+n,k}(t)S_{i,k}(t)+\dot{S}_{i+n,k+n}(t)S_{i,k+n}(t)\\  &-\sum_{i=1}^n\dot{S}_{i,k}(t)S_{i+n,k}(t)+\dot{S}_{i,k+n}(t)S_{i+n,k+n}(t)   \end{split}
\end{equation}
For $k\in\llbracket 1,n\rrbracket$ and $t\in\bbS^1$, let us denote by $U_k(t)=\begin{pmatrix}
\cos\theta_k(t)& -\sin\theta_k(t)\\
\sin\theta_k(t)& \cos\theta_k(t)
\end{pmatrix}$ the $k$-th diagonal block of $U(t)$. 
Then, for $j\in\llbracket 1,2n\rrbracket$, $k\in\llbracket 1,n\rrbracket$ and $t\in\bbS^1$:
\begin{equation}
\begin{pmatrix}
S_{j,k}(t)\\
S_{j,k+n}(t) 
\end{pmatrix}=\ ^tU_k(t)\begin{pmatrix}
S_{0,j,k}(t)\\
S_{0,j,k+n}(t) 
\end{pmatrix}
\end{equation}
Therefore: 
\begin{equation}\label{eq:lienS0}
\begin{pmatrix}
\dot{S}_{j,k}(t)\\
\dot{S}_{j,k+n}(t) 
\end{pmatrix}=\ ^tU_k(t)\begin{pmatrix}
\dot{S}_{0,j,k}(t)\\
\dot{S}_{0,j,k+n}(t) 
\end{pmatrix}+\ ^t\dot{U}_k(t)\begin{pmatrix}
S_{0,j,k}(t)\\
S_{0,j,k+n}(t) 
\end{pmatrix}
\end{equation}
Let us now observe that for $k\in\llbracket 1,n\rrbracket$, and any $t\in\bbS^1$:
\begin{equation}\label{eq:Ukdot}
\dot{U}_k(t)\ ^tU_k(t)=\begin{pmatrix} 0 &-1\\1&0\end{pmatrix}
\end{equation}
Therefore, since for any $k\in\llbracket 1,n\rrbracket$ and $t\in\bbS^1$ $U_k(t)$ is an orthogonal matrix and $S_0(t)$ is a symplectic matrix, we get from equations \eqref{eq:lienS0} and \eqref{eq:Ukdot}:\begin{equation}\label{eq:deterthetak} \begin{split}
\frac{\partial^2q_{S}}{\partial x_k^2}(t)+\frac{\partial^2q_{S}}{\partial \xi_k^2}(t)&=\frac{\partial^2q_{S_0}}{\partial x_k^2}(t)+\frac{\partial^2q_{S_0}}{\partial \xi_k^2}(t)\\
&+2\dot{\theta}_k(t)
\end{split}
\end{equation}
Since the function $t\mapsto \frac{\partial^2q_{S}}{\partial x_k^2}(t)+\frac{\partial^2q_{S}}{\partial \xi_k^2}(t)$ has been determined above, and the function $t\mapsto \frac{\partial^2q_{S_0}}{\partial x_k^2}(t)+\frac{\partial^2q_{S_0}}{\partial \xi_k^2}(t)$ is entirely determined by the explicitly constructed function $t\mapsto S_0$, equation \eqref{eq:deterthetak} then determines the function $\dot{\theta}_k$. Therefore, the function $t\mapsto U(t)$, hence the function $t\mapsto S^{\sigma}(t)$, is determined up to right multiplication by a constant block diagonal matrix $U_0$ whose diagonal block matrices are $2$ by $2$ rotations. It is now sufficient to observe that if two functions $t\mapsto S_1(t)$ and $t\mapsto S_2(t)$ are related by the equation: 
\begin{equation}
S_2^{\sigma}=S_1^{\sigma}U_0
\end{equation}
where $U_0$ is a constant matrix, then
\begin{equation}
\phi_{S_2}=\phi_{S_1}\circ\phi_{U_0^{\sigma^{-1}}}
\end{equation}
and, if $U_0$ is a constant block diagonal matrix  whose diagonal block matrices are $2$ by $2$ rotations:
\begin{equation}
H^{0}\circ\phi_{U_0^{\sigma^{-1}}}=H_0
\end{equation}
Finally, the choice of $U_0$ in the determination of $t\mapsto S(t)$ does not change the validity of equation \eqref{eq:H0phi} for $\phi=\phi_S$, and Theorem \ref{mainFermi} is proved. 
\end{proof}

\section{Classical analogs}\label{class}
In this section we prove a classical result, analog to our preceding quantum ones, and motivated by the following straightforward lemma. 

\begin{lemma}\label{four}
Let $O$ be a polynomial operator on $L^2(\R^n\times \bbS^1)$ whose Weyl symbol, expressed in polar and cylindrical coordinates is the function
$\mathcal O(A,\tau,\vp,\tau)$. Then
\be\label{four1}
\mul  O\mur=\int_{\mathbb T^n\times \bbS^1}\mathcal O(\mu\hbar,\nu\hbar,\vp,t)d\vp dt+O(\hbar).
\ee
where for any $j=1\dots n$, \ $x_j+i\xi_j=\sqrt{A_j}e^{i\vp_j}$.
\end{lemma}

We concatenate analogs of Theorems \ref{mainFermi} and \ref{main} in the following

\begin{theorem}\label{classmain} 
Let $\gamma$ be a non-degenerate elliptic periodic trajectory of the Hamiltonian flow generated by a proper smooth Hamiltonian function $H$. 
Let  $\mathcal\w_p^k$ be functions satisfying in a local symplectic system of coordinates $(x,t,\xi,\tau)\in  T^*(\R^n\times \bbS^1)$ such that $\gamma=\{x=\xi=\tau=0\}$:
\begin{equation}\mathcal\w^0_p(x,t,\xi,\tau)=e^{-2i\pi pt}\tau \ \text{ and }\ \mathcal\w^k_p(x,t,\xi,\tau)=e^{-2i\pi pt}\uv^k(x,\xi),\ k=1,\dots, 2n^2+n\end{equation}
 with the property that 
\textbf{$\mathcal\uv^k(0)=\nabla\mathcal\uv^k(0)=0$ and the  Hessians  $d^2\mathcal\uv^k(0)$ are linearly independent}.

Let $\Phi$ be the formal (unknown a priori) symplectomorphism which leads to the Birkhoff normal form near $\gamma$ and $(A,\tau,\vp,t)$ the corresponding 
(formal and also unknown a priori) action-angle coordinates such that $\gamma=\{A=\tau=0\}$. Let us define near $A=\tau=0$ the following ``average" quantities
\be\label{matel}
\overline{\mathcal \w_p^k}(A,\tau):=\int_{\mathbb T^n\times \bbS^1}\mathcal \w_p^k\circ\Phi(A,\tau,\vp,t)d\vp dt.
\ee
 
Then the knowledge of $\nabla\overline{\mathcal{\w}_p^k}(0,0)$ for $k=1,\dots,2n^2+n$,
determines (in a constructive way) an explicit  system of Fermi coordinates near $\gamma$.
\vskip 0.5cm

Moreover, let now $(x,t,\xi,\tau)\in  T^*(\R^n\times \bbS^1)$ be \textbf{any} system of Fermi coordinates near $\gamma$ and let for $(m,n,p)\in\bbN^{2n}\times\bbZ$, $\calO_{mnp}$, $\calO_p$ be functions satisfying in a neighborhood of $\gamma$:  
\be
\left\{\begin{array}{rcl}
\mathcal{O}_{mnp}(x,t,\xi,\tau)&=&e^{-i2\pi pt}\prod_{j=1}^n\left(\frac{x_j+i\xi_j}{\sqrt{2}}\right)^{m_j}\left(\frac{x_j-i\xi_j}{\sqrt{2}}\right)^{n_j}+O\left(\vert A\vert+\vert\tau\vert)^{\frac{\vert m\vert+\vert n\vert+1}{2}}\right)\\
\mathcal{O}_{p}(x,t,\xi,\tau)&=&e^{-i2\pi pt}\tau+O\left(\vert A\vert+\vert\tau\vert)^{\frac{3}{2}}\right)
\end{array}\right.
\ee

Then the knowledge of the Birkhoff normal form near $\gamma$ and of the Taylor expansion at $A=\tau=0$ up to order $N$ of  $\overline{\mathcal O}_{mnp}$, $\overline{\mathcal O}_{q}$, defined as in \eqref{matel}  for
\begin{enumerate}
\item $0<\vert m\vert+\vert n\vert\leq N$
\item  $\forall j=1\dots n,\ m_j=0$\ \textbf{or} $ n_j =0$ 
\item $p\in\bbZ$, $q\in\bbZ^*$
\end{enumerate}
determines the Taylor expansion of the ``true" Hamiltonian $H$ up to the same order in the picked-up system of Fermi coordinates.
\end{theorem}
\begin{proof}
Let us first prove the second part of Theorem \ref{classmain}.
Let $(x,t,\xi,\tau)\in  T^*(\R^n\times \bbS^1)$ be a system of Fermi coordinates near $\gamma$. Let, for $N\geq 3$, $F_N$ be the principal symbol of $W_N$. With the notations of Proposition \ref{W} we write 
\begin{equation}
F_N(z,t,\bz,\tau)=\sum_{\vert j\vert+\vert k\vert+2s=N}\alpha_{0jks}(t)e^{2i\pi pt} z^j\bz^k\tau^s
\end{equation}
 Let $F$ satisfy
\begin{equation}
F\sim \sum_{N=3}^{+\infty} F_N
\end{equation}
With the notations of the proof of Proposition \ref{W}, we have: 
\begin{equation}
H\circ\exp(\chi_{F})(x,t,\xi,\tau)\sim h(p,\tau,0)
\end{equation}

Let $N\geq 3$, $(m,n,p,q)\in\bbN^{2n}\times \bbZ\times\bbZ^*$
and $(j,k,s)\in\bbN^{2n+1}$ satisfy: 
\begin{equation}\label{conditionmnpq}
0<\vert m\vert+\vert n\vert\leq N,\ \   \ m_in_i=0\ \forall i\in\llbracket 1, n\rrbracket, \ \vert j\vert+\vert k\vert +2s=N
\end{equation}
Then, as it has been seen in the proof of Lemma \ref{equivalent}, if $\sigma^1_{jks}$ and $\sigma^2_{jks}$ are the symbols of $\{\alpha_{0jks}(t)e^{2i\pi pt} z^j\bz^k\tau^s,\mathcal{O}_{mnp}\}$ and $\{\alpha_{0jks}(t)e^{2i\pi pt} z^j\bz^k\tau^s,\mathcal{O}_{q}\}$ respectively, we have, if $j+m=k+n$,  
\begin{equation}\label{sigma11}\begin{split}
\int_{\mathbb T^n\times \bbS^1}\sigma^1_{jks}(A,\tau,\vp,t)d\vp dt=& ic_p(\alpha_{0jks})A^{\max(j,k)}\tau^{s}\left(\sum_{i=1}^n\frac{k_im_i-j_in_i}{A_i}+ \frac{2\pi ps}{\tau}\right)\\&+O\left(\vert A\vert+\vert\tau\vert)^{\frac{N+\vert m\vert +\vert n\vert -1}{2}}\right)
\end{split}
\end{equation}
and if $j+m\neq k+n$,
\begin{equation}\label{sigma12}
\int_{\mathbb T^n\times \bbS^1}\sigma^1_{jks}(A,\tau,\vp,t)d\vp dt=O\left(\vert A\vert+\vert\tau\vert)^{\frac{N+\vert m\vert +\vert n\vert -1}{2}}\right)
\end{equation}
We also have, if $j=k$:
\begin{equation}\label{sigma21}
\int_{\mathbb T^n\times \bbS^1}\sigma^2_{jks}(A,\tau,\vp,t)d\vp dt=ic_q(\alpha_{0jks})2\pi q(1+s) A^{j}\tau^s +O\left(\vert A\vert+\vert\tau\vert)^{\frac{N+1}{2}}\right)
\end{equation}
and if $j\neq k$:
\begin{equation}\label{sigma22}
\int_{\mathbb T^n\times \bbS^1}\sigma^2_{jks}(A,\tau,\vp,t)d\vp dt=O\left(\vert A\vert+\vert\tau\vert)^{\frac{N+1}{2}}\right)
\end{equation}

Now, let us set $F_2=0$ and assume that the function $F_{\leq N-1}$ has been determined for some $N\geq 3$. Then for $l\geq 2$, $\{\overbrace{F_{\leq N},\{\dots,\{F_{\leq N}}^{l\ \text{times}},\mathcal{O}_{mnp}\}\}\}$ and $\{\overbrace{F_{\leq N},\{\dots,\{F_{\leq N}}^{l\ \text{times}},\mathcal{O}_{q}\}\}\}$ are determined modulo $O\left(\vert A\vert+\vert\tau\vert)^{\frac{N+\vert m\vert+\vert n\vert-1}{2}}\right)$ and $O\left(\vert A\vert+\vert\tau\vert)^{\frac{N+1}{2}}\right)$ respectively. 

Therefore, by \eqref{sigma11} and \eqref{sigma12}, $\overline{\calO}_{mnp}(A,\tau)$ is equal, modulo $O\left(\vert A\vert+\vert\tau\vert)^{\frac{N+\vert m\vert+\vert n\vert-1}{2}}\right)$ and known terms to:
\begin{equation}
\sum\limits_{\substack{\vert j\vert+\vert k\vert+2s=N\\j=m=k+n}}ic_p(\alpha_{0jks})A^{\max(j,k)}\tau^{s}\left(\sum_{i=1}^n\frac{k_im_i-j_in_i}{A_i}+ \frac{2\pi ps}{\tau}\right)
\end{equation}
and by \eqref{sigma21} and \eqref{sigma22}, $\overline{\calO}_{q}(A,\tau)$ is equal, modulo known terms and $O\left(\vert A\vert+\vert\tau\vert)^{\frac{N+1}{2}}\right)$ to:
\begin{equation}
\sum\limits_{2\vert j\vert+2s=N}ic_q(\alpha_{0jjs})2\pi q(1+s) A^{j}\tau^s
\end{equation}
Thus, just as in the proof of Proposition \ref{propp}, let $(m,n,p,q)\in \bbN^{2n}\times\bbZ\times\bbZ^*$ run over all possible values under condition \eqref{conditionmnpq}, we determine every function $\alpha_{0jks}$, $\vert j\vert+\vert k\vert+2s=N$, hence $F_{N}$, which concludes the proof of the second part of Theorem \ref{classmain}.

The proof of the first part of the Theorem follows the same strategy with respect to Theorem \ref{mainFermi} than the proof of the second part with respect to Proposition \ref{propp}. 
\end{proof}

The last result of this paper will be the classical analog of Theorems \ref{corbot} and \ref{hope2}. 

Let us remind, \cite{gur}, that in the case where $H(x,\xi)=\xi^2+V(x)$  the  classical normal form determines  the Taylor expansion of the potential when the latter is invariant by the symmetry $x_i\to -x_i$ for each $i\in\llbracket1,n\rrbracket$. In the general case the Taylor expansion of the averages, in the sense of \eqref{matel}, of a finite number of classical observables are necessary to recover the full potential.

Let us assume $H\in\calC^{\infty}(T^*\calM,\bbR)$ has a global minimum at $z_0\in T^*\mathcal M$, and let $d^2H_p(z_0)$ be the Hessian of $H$ at $z_0$. Let us define the matrix $\Omega$ defined by  $d^2H_p(z_0)(\cdot, \cdot)=:\omega_{z_0}(\cdot, \Omega^{-1}\cdot)$ where $\omega_{z_0}(\cdot,\cdot)$ is the canonical symplectic form of $T^*\mathcal M$ at $z_0$. The eigenvalues of $\Omega$ being purely imaginary, we denote them by $\pm i\theta_j$ with $\theta_j>0, \ j\in\llbracket1,n\rrbracket$. Let us assume that $\theta_j, j\in\llbracket1,n\rrbracket$ are rationally independent.

\begin{theorem}\label{corbotclass}

The statement of Theorem \ref{classmain} remains valid verbatim by replacing $\gamma$ by $z_0$ and ignoring the variables $t,\tau$.
\end{theorem}
Let us now enunciate the classical analog of Theorem \ref{hope2} in the case of a Schrödinger operator with potential $V\in\calC^{\infty}(\calM,\bbR)$: 
\begin{theorem}\label{hopeclass}
Let $q_0$ be a global non-degenerate minimum of $V$ on $\mathcal M$. Let us assume that the square roots of the eigenvalues of $d^2V(q_0)$
 are linearly independent over the rationals.

Let $\mathcal\w^k, k=1\dots\frac{n(n+1)}2,$ be  smooth functions on $\mathcal M$ such that
\textbf{$\mathcal\w^k(q_0)=\nabla\mathcal\w^k(q_0)=0$ and the  Hessians  $d^2\mathcal\w^k(q_0)$ are linearly independent}
(an example of such potentials is the family $\mathcal Q_{ij}(x)=x_ix_j$ in a 
local system of coordinates centered at $q_0$).

Then the knowledge of the Birkhoff normal form near $q_0$ and of the Taylor expansion at $A=0$ up  of the  (finite number of) ``average"  $\overline{\mathcal\w^k}(A)$
determines (in a constructive way) an explicit system of Fermi coordinates.

\vskip 0.5cm
Moreover, 
 let $(x,\xi)\in T^*\R^n$ be any system of Fermi coordinates centered at $(q_0,0)$ and $\mathcal O_{m0}$ defined in Theorem \ref{hope2}.

Then the knowledge  of the Taylor expansion at $A=0$ up to order $N\geq 3$ of the  (finite number) ``average" quantities $\overline{\mathcal O_{m0}}$ as in 
\eqref{matel}

together with the Birkhoff normal form itself, determines the Taylor expansion up to order $N$ of $V$ at $q_0$ in the picked-up system of coordinates. 
 
\end{theorem}

 In the line of the proof of  Theorem \ref{classmain} the proofs of Theorem \ref{hopeclass} and \ref{corbotclass}  are easy adaptations of the proofs of Theorem \ref{hope2} and \ref{corbot}. We omit them here.

\vskip 0.6cm

\begin{appendix}

\section{Lemmas on linear and bilinear algebra}

\begin{lemma}\label{lem:diagsymp}
Let $q$ be a positive quadratic form on $\bbR^{2n}$. Then there exists a canonical endomorphism $\phi$ on $\bbR^{2n}$, and a $n$-tuple of positive real numbers $(\lambda_1,\dots,\lambda_n)$, defined as the imaginary part of the eigenvalues of positive imaginary part of the endomorphism defined by: 
\begin{equation}
\langle\cdot;a(\cdot)\rangle_q=\omega(\cdot,\cdot)
\end{equation}
where $\langle\cdot;\cdot\rangle_q$ be the scalar product associated to $q$ and $\omega$ the canonical symplectic form on $\bbR^{2n}$, and such that: 

\begin{equation}\label{eq:SAS}
\forall (x,\xi)\in\bbR^{2n}, \ q(\phi(x,\xi))=\sum_{i=1}^n \lambda_i (x_i^2+\xi_i^2)
\end{equation}
Moreover, if the real numbers $\lambda_1,\dots,\lambda_n$ are pairwise different, and $\phi'$ is an endomorphism of $\bbR^{2n}$. Then $\phi'$ is canonical and satisfies \eqref{eq:SAS} if and only there exists an orthogonal isomorphism $u$ on $\bbR^{2n}$ whose restriction to the plane spanned by $(\frac{\partial}{\partial x_i},\frac{\partial}{\partial \xi_i})$ (for any $i\in\llbracket 1,n\rrbracket$) is a rotation, such that $\phi'=\phi\circ u$. 

\end{lemma}

\begin{proof}[Proof of Lemma \ref{lem:diagsymp}]
$a$ is antisymmetric with respect to $q$, and therefore there exists a $q$-orthonormal basis of $\bbR^{2n}$ $(u_1,\dots,u_n,v_1,\dots,v_n)$ and a $n$-tuple of positive real numbers $(\lambda_1,\dots,\lambda_n)$ such that, for $j\in\llbracket 1,n\rrbracket$: 
\begin{equation}
\lambda_ja(u_j)=- v_j \text{ and } \lambda_j a(v_j)=u_j
\end{equation}

Now let us set, for $j\in\llbracket 1,n\rrbracket$: 
\begin{equation}
\tilde{u}_j=\sqrt{\lambda_j}u_j \text{ and } \tilde{v}_j=\sqrt{\lambda_j}v_j
\end{equation}

Then, $(\tilde{u}_1,\dots,\tilde{u}_n,\tilde{v}_1,\dots,\tilde{v}_n)$ is a $q$-orthogonal basis of $\bbR^{2n}$ satisfying, for $j\in\llbracket 1,n\rrbracket$, $q(\tilde{u}_j)=\lambda_j$ and $q(\tilde{v}_j)=\lambda_j$, and the preceding properties together with \eqref{eq:SAS} implies that it is also a symplectic basis, which concludes the proof of the first part of Lemma \ref{lem:diagsymp}.

To prove the second part of Lemma \ref{lem:diagsymp}, let us consider another symplectic and orthogonal basis $(u'_{1},\dots,u'_n,v'_1,\dots,v'_n)$ where, for  $j\in\llbracket 1,n\rrbracket$, the $q$-norm of $u'_j$ and $v'_j$ is $\lambda_j$. 
Then, by \eqref{eq:SAS}, for any $j\in\llbracket 1,n\rrbracket$, $a(u'_j)$ is orthogonal to any vector of the basis but $v'_j$ and $\langle v'_j,a(u'_j)\rangle_q=w(v'_j,u'_j)=-1$, therefore $\lambda_j a(u'_j)=-v'_j$, and by the same argument, $\lambda_j a(v'_j)=u'_j$.

Therefore, the plane spanned by $(u_j,v_j)$ and the plane by $(u'_j,v'_j)$  are both the kernel of $a^2+\lambda_j^2$ (2-dimensional since we made the additional assumption the $\lambda_i$s are pairwise different). Therefore, if $\phi$ and $\phi'$ are the endomorphisms which send the canonical basis of $\bbR^{2n}$ to basis $(\tilde{u}_1,\dots,\tilde{u}_n,\tilde{v}_1,\dots,\tilde{v}_n)$ and basis $(u'_{1},\dots,u'_n,v'_1,\dots,v'_n)$ respectively, then one can considerer the restriction to any plane spanned by $(\frac{\partial}{\partial x_i},\frac{\partial}{\partial \xi_i})$ (for any $i\in\llbracket 1,n\rrbracket$) is an orthogonal symplectomorphism from the plane to itself, that is a rotation. 

\end{proof}

Let $\sigma$ be the permutation of $\llbracket 1,2n\rrbracket$ defined by: 
\begin{equation}\label{eq:defsigma}
\forall i\in\llbracket 1,2n\rrbracket,\ \sigma(i)=\left\{\begin{array}{ll}
2i-1\ &\mbox{ si }  i\leq n\\
2(i-n)\ &\mbox { si } i\geq n+1 
\end{array}\right.
\end{equation}
and $M_\sigma$ be the associated permutation matrix (\emph{i.e.} for any $(i,j)\in\llbracket 1,2n\rrbracket^2$, $(M_\sigma)_{ij}=\delta_{\sigma(i),j}$.

Now, let us set, for any matrix $S\in\mathcal{M}_{2n}(\bbR)$: 
\begin{equation}\label{eq:defSsigma}
S_{\sigma}=M_{\sigma}^{-1}SM_{\sigma}.
\end{equation} 

Let us also, for $(i,k)\in\llbracket 1,2n\rrbracket\times \llbracket 1,n\rrbracket$,  denote by $L_{S,i,k}$ the vector of $\bbR^2$ defined by $L_{S,i,k}=\left(\begin{array}{ll}
(S_\sigma)_{i,2k-1}\\(S_\sigma)_{i,2k}
\end{array}\right)\in\bbR^2$. Then, for  $(i,k)\in\llbracket 1,n\rrbracket^2$, $\mathfrak{s}_{i,k}$ will be the  matrix  of size $2$ whose first line is $^tL_{S,2i-1,k}$ and second line $^t L_{S,2i,k}$.

\begin{lemma}\label{lem:invariantS} Let $A$ be a positive matrix of size $2n$. 
Let $\mathcal{S}$ be the  (non-empty by lemma \ref{lem:diagsymp}) set of symplectic matrices satisfying \begin{equation}\label{eq:SHS}
\ ^tS A S=\left(\begin{array}{c|c}
D_{\lambda}&0\\ \hline
0&D_{\lambda}
\end{array}\right)
\end{equation}
where $D_{\lambda}$ is the diagonal matrix with $(\lambda_1,\dots,\lambda_n)$ as $n$-tuple of positive diagonal elements, which we assume pairwise different. 
Then: 
\begin{enumerate}
\item The family $(\langle L_{S,i,k};L_{S,j,k}\rangle)_{i,j)\in\llbracket 1,2n\rrbracket^2,k\in\llbracket 1,n\rrbracket}$ is independent of matrix $S\in\mathcal{S}$. 

\item Once the preceding invariants of  $\mathcal{S}$ given, one can construct explicitly a particular matrix of $\mathcal{S}$ (hence all of them by Lemma \ref{lem:diagsymp}).
\end{enumerate}
\end{lemma}

\begin{proof}[Proof of Lemma \ref{lem:invariantS}]
Let us first prove the first point.  
Let $(S,T)\in\mathcal{S}^2$. By Lemma \ref{lem:diagsymp}, there exist $n$ matrices belonging to $SO_2(\bbR)$ and denoted by $O_1,\dots,O_n$, such that: 
\begin{equation}\label{eq:StildeS}
T_{\sigma}=S_{\sigma}\left(\begin{array}{ccc}
 O_1   & &\\ 
&\ddots & \\
& & O_n \end{array}\right) =\left(\begin{array}{ccc}
 \mathfrak{s}_{1,1}O_1   &\cdots &\mathfrak{s}_{1,n}O_n\\ 
\vdots & & \vdots \\
\mathfrak{s}_{n,1}O_1 & \cdots & \mathfrak{s}_{n,n}O_n \end{array}\right)
\end{equation}  

and \eqref{eq:StildeS} is equivalent to: 
\begin{equation}
\forall (i,k)\in\llbracket 1,2n\rrbracket\times \llbracket 1,n\rrbracket, L_{T,i,k}=\ ^t O_k L_{S,i,k}
\end{equation}

Hence  $(\langle L_{S,i,k};L_{S,j,k}\rangle)_{i,j)\in\llbracket 1,2n\rrbracket^2,k\in\llbracket 1,n\rrbracket}$
does not depend on the matrix $S\in\mathcal{S}$ and the first point of Lemma \ref{lem:invariantS} is proven.  

Now, let $S\in\mathcal{S}$, and let $(a_{ijk})_{(i,j)\in\llbracket 1,2n\rrbracket^2,k\in\llbracket 1,n\rrbracket}$ be the family defined by: 
\begin{equation}
\forall (i,j)\in\llbracket 1,2n\rrbracket^2, \ \forall k\in\llbracket 1,n\rrbracket, a_{ijk}=\langle L_{S,i,k};L_{S,j,k}\rangle
\end{equation}
Let us assume that this family is given. 
Two vectors $u$ and $v$ of $\bbR^2$ are independent if and only if: $\langle u;v\rangle^2 < \langle u;u\rangle \langle v;v\rangle$.
Since matrix $S$ is invertible, on can choose, for any $k\in\llbracket 1,n\rrbracket$, a couple of indices $(i_k,j_k)\in \llbracket 1,2n\rrbracket^2$ such that: 
\begin{equation}\label{eq:lib}
a_{i_kj_kk}^2 < a_{i_ki_kk} a_{j_kj_kk} 
\end{equation}
Let $k\in\llbracket 1,n\rrbracket$
Let us choose  a vector $v_{i_{k}k}$, whose norm is $\sqrt{a_{i_{k}i_kk}}>0$.
The following system of equations with unknown $v\in\bbR^2$: \begin{equation}\label{eq:detervjkk}
\left\{\begin{array}{lll}
\langle v_{i_kk};v\rangle&=a_{i_kj_kk}\\
\langle v;v\rangle&=a_{j_kj_kk} 
\end{array}\right.
\end{equation}
admits exactly two solutions (by \eqref{eq:lib}), denoted by $v^{+}_{j_kk}$ et $v^{-}_{j_kk}$ obtained from one another by orthogonal symmetry $R_k$ of axis the line spanned by $v_{i_kk}$.

Let us set $v^{-}_{i_kk}=v^{+}_{i_kk}=v_{i_kk}$. 
Since the families $(v^{+}_{i_kk},v^{+}_{j_kk})$ et  $(v^{-}_{i_kk},v^{-}_{j_kk})$ are two basis of $\bbR^2$, for any $i\in\llbracket 1,2n\rrbracket\setminus \{i_k,j_k\}$, each one of the two systems: 
\begin{equation}\label{eq:detervik}
\left\{\begin{array}{ll}
\langle v_{i_kk};v\rangle&=a_{i_kik}\\
\langle v^{+}_{j_kk};v\rangle&=a_{j_kik} 
\end{array}\right.  \ \text{ et }\ \left\{\begin{array}{ll}
\langle v_{i_kk};v\rangle&=a_{i_kik}\\
\langle v^{-}_{j_kk};v\rangle&=a_{j_kik} 
\end{array}\right. 
\end{equation}
admits exactly one solution denoted respectively by $v^{+}_{ik}$ and $v^{-}_{ik}$, and satisfying relation $v^{-}_{ik}=R_kv^{+}_{ik}$.

We are now able to construct $2^n$ matrices $(T_A)_{A\in\mathcal{P}(\llbracket 1,n\rrbracket)}$ defined, for $A\in\mathcal{P}(\llbracket 1,n\rrbracket)$, by: 
\begin{equation}\label{eq:defTA}
\forall (i,k)\in\llbracket 1,2n\rrbracket\times \llbracket 1,n\rrbracket, L_{T_A,i,k}=\left\{\begin{array}{ll}
v^{+}_{ik}\ \text{ if } \ k\in A\\
v^{-}_{ik} \ \text{ if else}
\end{array}\right. 
\end{equation}

In order to prove the second point of Lemma \ref{lem:invariantS}, it is sufficient to prove the following assertions: 

\begin{enumerate}
\item There exists at least one set $A\in\calP(\llbracket 1,n\rrbracket)$, such that: $T_A\in \calS$. 
\item There is at most one set $A\in\calP(\llbracket 1,n\rrbracket)$, such that $T_A$ is symplectic (and $A$ is determined by family $(a_{ijk})_{(i,j)\in\llbracket 1,2n\rrbracket^2,k\in\llbracket 1,n\rrbracket}$)
\end{enumerate}

Indeed, once those two assertions proved, there will be exactly one set $A\in\calP(\llbracket 1,n\rrbracket)$ such that $T_A$ is symplectic, and it will be an element of $\calS$, constructed from the values of family $(a_{ijk})_{(i,j)\in\llbracket 1,2n\rrbracket^2,k\in\llbracket 1,n\rrbracket}$ only.

Let us prove the first assertion. Let, for any $k\in\llbracket 1,n\rrbracket$, $O_k$ be the unique element of $SO_2(\bbR)$ such that $L_{S,i_k,k}=O_k v_{i_kk}$ (where $S$ is a particular matrix of $\calS$).

The system \eqref{eq:detervjkk} is equivalent to: 

\begin{equation}
\left\{\begin{array}{lll}
\langle L_{S,i_k,k} ;O_k v\rangle&=a_{i_kj_kk}\\
\langle O_k v;O_k v\rangle&=a_{j_kj_kk} 
\end{array}\right.
\end{equation}

which admits exactly two solutions: $v^{+}_{j_kk}$ et $v^{-}_{j_kk}$. Hence, for any $k\in\llbracket 1,n\rrbracket$: 

\begin{equation}
L_{S,j_k,k}=O_kv^{+}_{j_kk} \ \text{ or} \ L_{S,j_k,k}=O_kv^{-}_{j_kk}
\end{equation}

Let us define the set $A$ by: \begin{equation}A=\{k\in\bbN\ \vert L_{S,j_k,k}=O_kv^{+}_{j_kk}\}\end{equation}
Since each system \eqref{eq:detervik} admit a unique solution, we obtain:  

\begin{equation}\begin{split}
\forall (i,k)\in\llbracket 1,2n\rrbracket\times \llbracket 1,n\rrbracket, L_{S,i,k}&=\left\{\begin{array}{ll}
O_k v^{+}_{ik}\ \text{ if } \ k\in A\\
O_k v^{-}_{ik} \ \text{ if else}
\end{array}\right. \\
&=O_kL_{T_{A},i,k}\end{split}
\end{equation}
that is: 
\begin{equation}
T_{A,\sigma}=S_{\sigma}\left(\begin{array}{ccc}
 O_1   & &\\ 
&\ddots & \\
& & O_n \end{array}\right) 
\end{equation}
and $T_A\in\calS$ by Lemma \ref{lem:diagsymp}. 

In order to prove the second assertion, let us use the following lemma: 
\begin{lemma}\label{lem:sympconstr}
For any symplectic matrix $B$ of size $2n$, we have: 
\begin{equation}
\forall k\in\llbracket 1,n\rrbracket, \ \sum_{i=1}^n \det (\mathfrak{b}_{i,k})=1
\end{equation}
\end{lemma}
If $A_1$ and $A_2$ are two parts of $\llbracket 1,n\rrbracket$, we get from \eqref{eq:defTA} and  relation $v^{-}_{ik}=R_kv^{+}_{ik}$ that: 
\begin{equation}
\forall (i,k)\in\llbracket 1,2n\rrbracket\times \llbracket 1,n\rrbracket, L_{T_{A_2},i,k}=\left\{\begin{array}{ll}
R_k L_{T_{A_1},i,k}\ &\text{ if } \ k\in A_1\Delta A_2\\
 L_{T_{A_1},i,k} \ &\text{ if else}
\end{array}\right. 
\end{equation} 
where $A_1\Delta A_2$ is the symmetric difference of $A_1$ and $A_2$: $A_1\Delta A_2=(A_1\setminus  A_2)\cup (A_2\setminus A_1)$. Hence: 

\begin{equation}
\forall k\in\llbracket 1,n\rrbracket, \ \sum_{i=1}^n \det \left((\mathfrak{t}_{A_2})_{i,k}\right)=\epsilon_k\sum_{i=1}^n \det \left((\mathfrak{t}_{A_1})_{i,k}\right)
\end{equation}
where, for $k\in\llbracket 1,n\rrbracket$, $\epsilon_k=-1$ if $k\in A_1\Delta A_2$, $\epsilon_k=1$ if else. Since $A_1\Delta A_2=\emptyset$ if and only if $A_1=A_2$, there exists at most one part $A$ of $\llbracket 1,n\rrbracket$ such $T_A$ is symplectic. 
The second assertion, hence the second point of Lemma \ref{lem:sympconstr}, is proven.
\end{proof}

\begin{proof}[Proof of Lemma \ref{lem:sympconstr}]
Since $B$ is a symplectic matrix, matrix $B_{\sigma}$ satisfies: 
\begin{equation}\label{eq:signeL} 
^tB_{\sigma}J_\sigma B_{\sigma}=J_\sigma
\end{equation}

It is sufficient, for $k\in\llbracket 1,n\rrbracket$, to read equality \eqref{eq:signeL} at line $2k$ and column  $2k-1$ to obtain: 

\begin{equation}
\sum_{i=1}^n \det (\mathfrak{b}_{i,k})=1
\end{equation}
\end{proof}

\begin{lemma}\label{lem:diagortho}
Let $A\in\calM_n(\bbR)$ be a positive matrix whose eigenvalues are pairwise different.
Let $D$ a diagonal matrix, similar to $A$. Then there exists exactly $2^n$ orthogonal matrices conjugating $A$ to $D$,
 and they are obtained  one from another by  a possible  change of the sign of each column. 
\end{lemma}

\begin{proof}[Proof of Lemma \ref{lem:diagortho}]
 As $A$ is positive, there exists an orthogonal matrix $Q_1$ such that: 
\begin{equation}
Q_1^{-1}AQ_1=\ ^tQ_1AQ_1=D
\end{equation}
Let $Q_2\in GL_n(\bbR)$. Then $Q_2$ is orthogonal and satisfies: $Q_2^{-1}AQ_2=D$ if and only if  $Q_2^{-1}Q_1$ is an orthogonal matrix which commutes to $D$, that is, because the diagonal elements of $D$ are pairwise different,  if and only if $Q_2^{-1}Q_1$ is an orthogonal diagonal matrix. Finally, $Q_2$ is orthogonal and satisfies: $Q_2^{-1}AQ_2=D$ if and only if $Q_2^{-1}Q_1$ is diagonal and its elements belong to $\{-1,1\}$, that is if  $Q_2$ is obtained from $Q_1$ by a possible change of the sign of each column.  
\end{proof}

\tcb{\section{Realizing the Poincaré angles}\label{deuxpi}}

\subsection{The periodic trajectory case}
In this section we indicate how different systems of Fermi coordinates and different Birkhoff normal forms exist for any realization of the Poincaré angles as real numbers and we show how those normal forms are linked to each other. Thus, our results are independent of this ambiguity.
\begin{proposition}\label{periodcbot}
Under the hypothesis of Theorem \ref{main},  the knowledge of the coefficients of the trace formula determines the quantities $e^{i\theta_i},i=1\dots n$. Moreover, let us denote by $\mathcal B_{\theta_1,\dots,\theta_n}(x,\xi,\tau)= E+\sum_{i=1}^n\theta_i \frac{x_i^2+\xi_i^2}{2}+\tau+O((x^2+\xi^2+\vert \tau\vert)^2)$ the Birkhoff normal form of $H_p$ associated to a given choice of angles $\theta_i$. For $k\in\bbZ^n$, let $h_k(x,\xi)=\sum \pi k_i(x_i^2+\xi_i^2)$ and let $\Phi_k$ be the symplectomorphism defined, with the notation of \eqref{chi}, by
\be
\Phi_k(x,\xi,t,\tau)=\left(\exp(t\chi_{h_k})(x,\xi),t,\tau+\pi\sum k_i(x_i^2+\xi_i^2)\right)
\ee
Then  $\mathcal B_{\theta_1,\dots,\theta_n}\circ\Phi=\mathcal B_{\theta_1+2k_1\pi,\dots,\theta_n+2k_n\pi}$
\end{proposition}
\begin{proof}
The first part of the assertion belongs to Fried \cite{fried}. The fact that $\Phi_k$ is a symplectomorphism can be checked directly. Moreover one sees immediately that it conjugates the quadratic in $(x,\xi)$/linear in $\tau$ part of  $\mathcal B_{\theta_1,\dots,\theta_n}$ to the one of $\mathcal B_{\theta_1+2k_1\pi,\dots,\theta_n+2k_n\pi}$. 
Moreover $\mathcal B_{\theta_1,\dots,\theta_n}\circ\Phi_k$ is a function of $\tau$ and $x_i^2+\xi_i^2$ only and it is easy to verify that the algorithmic constructions of the two normal forms are covariantly conjugated by $\Phi_k$.
Therefore, it is equal to $\mathcal B_{\theta_1+2k_1\pi,\dots,\theta_n+2k_n\pi}$.
\end{proof}
\subsection{The ``bottom of the well" case}
\begin{proposition}\label{poincbot}
Under the hypothesis of Theorems \ref{corbot} and \ref{hope2}, the knowledge of the spectrum of $H(x,\hbar D_x)$ in $[H_p(z_0),H_p(z_0)+\epsilon],\ \hbar=o(\epsilon),$ determine the $\theta_i$s up to permutation.
\end{proposition}
\begin{proof}
By the quantum normal Birkhoff form construction we know that the bottom part of the spectrum is $\{H_p(z_0)+\sum\limits_{i=1}^n\theta_i(\mu_i+1/2)\hbar+O(\hbar^2),
 |\mu\hbar|=O(\epsilon)\}$. Therefore, the bottom part of the spectrum determines the set $\Lambda=\{\sum\limits_{i=1}^n\theta_i(\mu_i+1/2), \mu\in\bbN^n\}$. 
Let us now assume that the $\theta_i$s are arranged in increasing order.  $\theta_1/2$ is then equal to the minimum of $\Lambda$. By induction, if $\theta_1,\dots,\theta_k$ are known for some $k$, $1\leq k<n$, let us define $\Lambda_k=\{\sum_{i=1}^k\theta_i(\mu_i+1/2)\}$. Let us set $\lambda_{k+1}:=\min\Lambda\cap\Lambda_k^{c}$. We easily see that $\theta_{k+1}=2\lambda_{k+1}-\sum_{i=1}^k\theta_i$, which concludes the proof. 
\end{proof}

\end{appendix}

\end{document}